\documentclass[amsmath,amssymb,twocolumn,preprintnumbers,superscriptaddress,showkeys]{revtex4-1}

\usepackage{color}
\usepackage{xcolor}
\usepackage{graphicx}
\usepackage{dcolumn}
\usepackage{bbm}
\usepackage{amsmath}
\usepackage{amsthm}
\usepackage{dsfont}
\usepackage{array}
\usepackage{wrapfig}
\usepackage{hyperref}
\usepackage[all]{hypcap}
\usepackage[latin10]{inputenc}
\usepackage{amsfonts}
\usepackage[english]{babel}
\usepackage{epstopdf}
\usepackage[cal=boondox]{mathalfa}

\newcommand{\bra}[1]{\langle #1|}
\newcommand{\ket}[1]{|#1\rangle}
\newcommand{\braket}[2]{\langle #1|#2\rangle}
\newcommand{\ketbra}[2]{|#1\rangle\!\langle#2|}
\newcommand{\id}{{\mathbbm 1}}

\newcommand{\mdag}{^{\dag}}

\newcounter{theorems}
\newtheorem{thm}[theorems]{Theorem}
\newtheorem{prop}[theorems]{Proposition}
\newtheorem{cor}[theorems]{Corollary}

\newtheorem{lem}[theorems]{Lemma}
\newtheorem{rem}[theorems]{Remark}
\newtheorem{ex.}{Example}[theorems]

\newtheorem*{cor*}{Corollary}
\newtheorem*{thm*}{Theorem}
\newtheorem*{prop*}{Proposition}
\newtheorem*{lem*}{Lemma}
\newtheorem*{rem*}{Remark}

\DeclareMathOperator{\tr}{Tr}
\DeclareMathOperator{\cut}{cut}

\DeclareMathOperator{\im}{i}

\DeclareMathOperator{\IC}{\begin{array}{c}
		\rightsquigarrow \\[-1.6 ex]
		\leftarrow
\end{array}}
\usepackage{txfonts}
\DeclareMathOperator{\CI}{\begin{array}{c}
		\leftsquigarrow \\[-1.6 ex]
		\rightarrow
\end{array}}

\begin{document}
\title{Of Local Operations and Physical Wires}
\author{Dario Egloff}
\affiliation{Institute of Theoretical Physics, Universit{\"a}t Ulm, Albert-Einstein-Allee 11D-89069 Ulm, Germany}
\author{Juan M. Matera}
\affiliation{IFLP-CONICET, Departamento de F{\'i}sica, Facultad de Ciencias Exactas, Universidad Nacional de La Plata, C.C. 67, La Plata 1900, Argentina}
\author{Thomas Theurer}
\affiliation{Institute of Theoretical Physics, Universit{\"a}t Ulm, Albert-Einstein-Allee 11D-89069 Ulm, Germany}
\author{Martin B. Plenio}
\affiliation{Institute of Theoretical Physics, Universit{\"a}t Ulm, Albert-Einstein-Allee 11D-89069 Ulm, Germany}

\begin{abstract}
	
In this work (multipartite) entanglement, discord and coherence are unified as different aspects of a single underlying resource theory defined through simple and operationally meaningful elemental operations. This is achieved by revisiting the resource theory defining entanglement, Local Operations and Classical Communication (LOCC), placing the focus on the underlying quantum nature of the communication channels. 
Taking the natural elemental operations in the resulting generalization of LOCC yields a resource theory that singles out coherence in the wire connecting the spatially separated systems as an operationally useful resource. The approach naturally allows to consider a reduced setting as well, namely the one with only the wire connected to a single quantum system, 
which leads to discord-like resources. The general form of free operations in this latter setting is derived and presented as a closed form. We discuss in what sense the present approach defines a resource theory of quantum discord and in which situations such an interpretation is sound -- and why in general discord is not a resource.
This unified and operationally meaningful approach makes transparent many features of entanglement that in LOCC might seem surprising, such as
the possibility to use a particle to entangle two parties, without it ever being entangled with either of them, or that there exist different forms of multipartite entanglement.
\end{abstract}

\maketitle

\noindent
One of the oldest questions in the field of quantum mechanics asks in what consists the difference between classical and quantum states. While certainly it is hard to compare frameworks that are so fundamentally different in nature,the main qualitative difference on a formal level is that quantum mechanics deals with probability amplitudes instead of probabilities. As a consequence, one of the main predictions of quantum mechanics is that physical systems can exhibit coherent superpositions of those states associated to  sharply defined values of the observables~\cite{Schr.1935} which can for instance be observed in interference experiments. It is mainly this difference that prevents quantum theory from
being explicable by a deterministic hidden variable model with variables that are local~\cite{EPR.1935,Bell.66,Asp.81,Hensen2015,Clauser1969} or non-contextual~\cite{Kochen1975,Mermin1993,Amselem2009}. 
The question then becomes how one can understand this difference in detail and furthermore how one can quantify it. This is relevant in  its own right and also to give solid foundations to the debates about how quantum mechanical observed coherences in biological systems~\cite{Huelga2013} are, or to objectively compare different platforms for quantum computers by measuring how much more resources they provide than their classical counterparts.

A particularly transparent approach
is given by the theory of local operations and classical communication (LOCC) (see~\cite{Plenio2007,Horodecki2009}, for reviews), which incorporates the idea of the impossibility of creating non-local superpositions (entanglement) if one has two distant parties that can only communicate classically, akin to Bell's argument for the non-classicality of quantum physics~\cite{Bell.66} (see section~\ref{sec:LOCCwiring}).

\begin{figure}[htb]
	\label{fig1}
	\includegraphics[width=8cm]{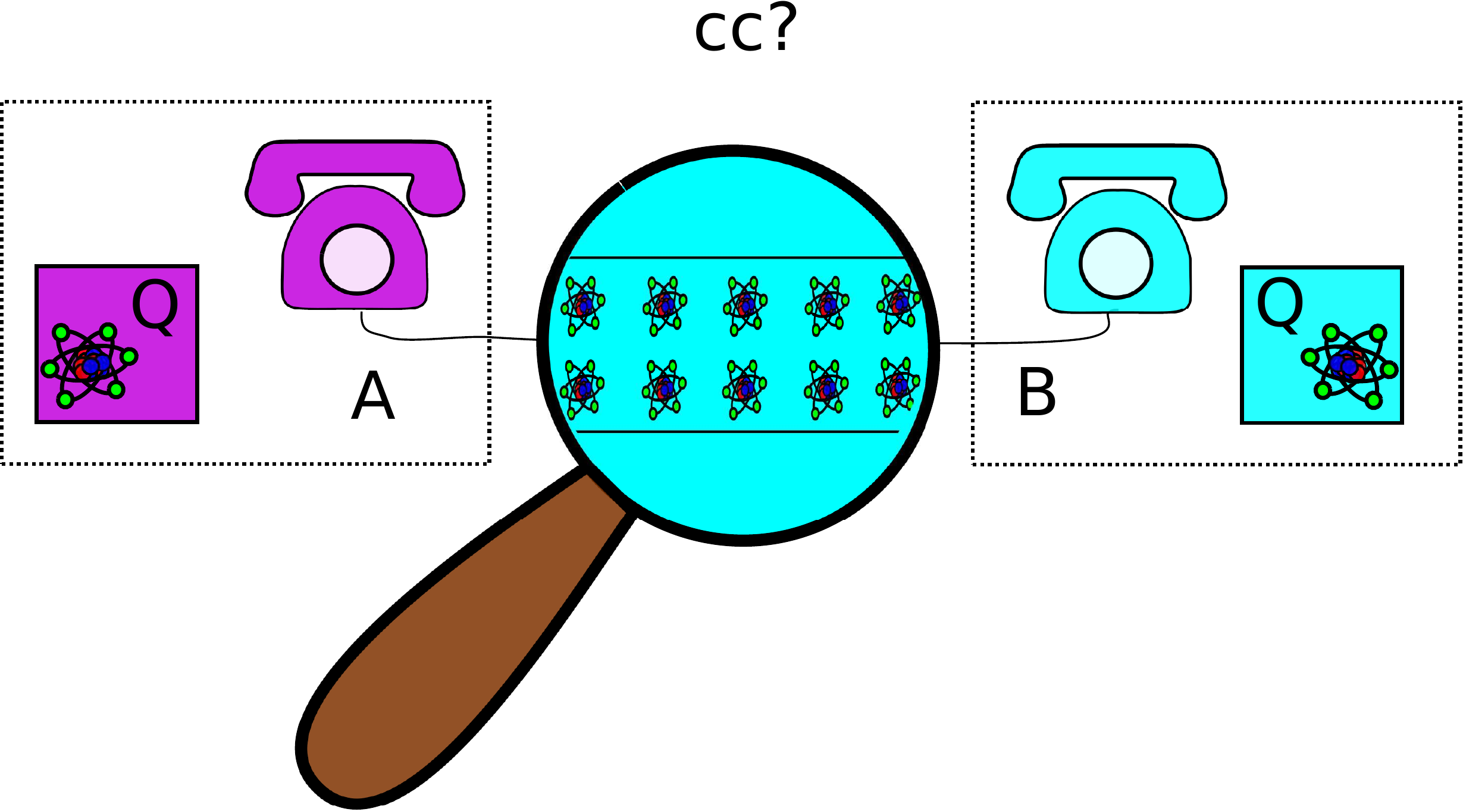}
	\caption{Wires for classical communication are quantum systems. }
\end{figure}

However, entanglement is not the only non-classical feature of quantum mechanics. There are other forms of superposition that can also give advantages over classical states. An instance in which this became apparent is the protocol of deterministic quantum computing with one qubit (DQC1) that can outperform any known classical algorithms even if no bipartite entanglement is present~\cite{Parker2002,Datta2005,DSC.08,Lanyon2008,Ali2014}. It was argued that a property denoted as discord~\cite{OZ.01,HV.01,Modi2012} would be the resource responsible for this operational advantage. While discord is an interesting measure, it is hard to argue that it is a resource of non-classicality, since it can be created already by mixing discord-free states (in fact product states) and mixing states amounts to forgetting information -- a task that is easily accomplished classically~\cite{OZ.01}. Some advances in understanding the resources in DQC1 have recently been made using a different form of non-classicality, called coherence,  which is exhibited by superpositions of states in a fixed orthonormal basis, whose elements and their statistical mixtures are the incoherent states ~\cite{Aberg2007,BCP.13,Winter2016,Streltsov2017b}. Based on this theory, complementing recently studied connections between entanglement and coherence~\cite{Streltsov2015b,Yao2015,Chitambar2016a,Chitambar2016c,Streltsov2016a,Hu2016,Hu2016a,Hu2017,Yadin2016,Zhu2017,Radhakrishnan2016,Streltsov2017,Girolami2017,Bu2017} it was shown that the precision of the DQC1 protocol is a function of the coherence of the qubit one uses as a control and that any state with some property called basis depended discord is a resource in this setting ~\cite{Ma2016,MEKP.16}. 
      
The standard (basis independent) discord  is recovered when a minimization  over this quantity is considered~\cite{MEKP.16}.
Ironically, in the theories used in these papers coherent states and entangled states are interconvertible resources and optimal instances of basis depended discord, making this more general variant of superposition in some sense equivalent to entanglement again. 
Even more generally, there is an operation that maps any superposition of a given set of linearly independent states (not necessarily orthogonal) to entangled states and non-superposition states to separable states~\cite{Killoran2016,Theurer2017,Regula2017}. This means that while entanglement does not seem to be the most generic non-classical resource, as it does not easily encompass all forms of superposition that seem useful for quantum tasks, the resources that enable to produce entanglement might well be.
Seemingly, the ability to do something truly quantum entails the possibility to produce entanglement. It is this ``universal character"~\cite{Sperling2017} of entanglement, that motivates the present paper.

Starting from LOCC, we treat the wire needed to do classical communication as a full physical (quantum) system, which can in principle realize states with some degree of coherence. We will show how the flow of this coherence can be used to create entanglement shared among the ends
of the wire (see Fig.~\ref{fig1}). We find a set of natural elemental operations 
on system and wires -- designated free in this setting -- that are equivalent to LOCC when wires are in incoherent states but, if not, allow to convert coherence into entanglement between the systems connected by the wire. In this setting, coherence in the wires is exactly as useful as entanglement between the two quantum systems, similarly as one might expect from a similar relationship between different forms of quantum cryptography (see e.g.~\cite{Shor2000,Horodecki2005a}). This setting thus explains the recently studied relation between these different resources~\cite{Streltsov2015b,MEKP.16,Yao2015,Chitambar2016a,Chitambar2016c,Streltsov2016a,Hu2016,Hu2016a,Hu2017,Ma2016,Yadin2016,Zhu2017,Radhakrishnan2016,Streltsov2017,Girolami2017,Bu2017} as the interplay of different facets of the same resource theory. Additionally, the present framework reduces to an effective theory of basis dependent quantum discord  in the case of a wire connected with just one quantum system. We will argue in what sense one can see it as a theoretical framework of quantum discord, and also discuss in some detail why -- even so -- discord cannot be called a resource. Indeed it even makes sense to focus on the effective theory on the wire alone, yielding a theory of coherence which we show to be very similar -- though not identical -- to the theory of coherence defined in~\cite{BCP.13}, but being motivated operationally instead of abstractly. 
Finally we generalize the setting to the multipartite case and show how in the tripartite case the non-equivalence of the $W$ and the $GHZ$ state (see~\cite{DVC.00}) can be understood operationally. All proofs are in the appendix.

\section{Local operations and physical wires}\label{sec:LOCCwiring}
As explained in the introduction, we are aiming at better understanding and quantifying non-classicality.
The tool of choice to develop this understanding are resource theories, since they provide a systematic guide for analysing situations where one wants to find and/or quantify the properties that can be useful for some tasks.
	
Abstractly, and glossing over details, to get a resource theory one puts a meaningful, but artificial restriction to what operations are allowed (free), within a given framework. The restriction should be chosen such that the connection to the property one wants to focus on is as clear as possible (which does not need to coincide with the distinction between ``easy" and ``hard"). There may be some preparations that are free operations and accordingly some states that are free. Since the states that are not free cannot be prepared by free operations, in some cases they might help to do an operation that otherwise would not be free, these are the resource states. There are more useful states and less useful states (if from one state one can reach another one with free operations the first is more useful, since it can be used for anything the second can be), imposing a partial order on the states. A measure for the resource can therefore only be meaningful if it is monotonous under the free operations, restricting strongly possible candidates.

One of the most successful resource theories in quantum information, and the starting point of our considerations, is the theory of local operations and classical communication (LOCC). This theory aims to capture the idea of EPR and Bell that in classical physics it is not possible to reproduce the effect of having non-local superpositions of states~\cite{EPR.1935,Bell.66}.
LOCC can be described by its elemental operations consisting of arbitrary local quantum operations on one system, post-selection and classical broadcasting of the result. Any LOCC operation is a concatenation of such operations (potentially depending on the results of the previous ones). Here we want to treat the broadcasting as an operation using a physical wire, instead of how it is usually done as an implicit exchange of classical information.  The broadcasting is then simply given by forwarding the state of the wire as an ancilla to the party in question. 

In this way, it makes sense to talk about the state of the communication channel as a quantum state. The standard LOCC theory is recovered by assuming that the state of the wire is forced to be incoherent, i.e,   $\rho_{wires}=\sum_i p(i) \ketbra{i}{i}$, and in a product state with both parties, meaning that the probability distribution is encoded in the diagonals. 
Note that while the basis one uses to encode a probability distribution in the wire is in principle arbitrary, one needs to fix it in advance; to be precise,  while it may change in time, this change must not depend on the measurement outcomes of the protocol, but must be defined at the beginning of the protocol.
Henceforth we will call this basis incoherent and denote it by $\cal Z$. With this definition it also makes sense to allow some classical processing in the wire as a free operation on this extended theory, that is, permutations of these basis states.
As an example of why it is important to fix the basis, and therefore which states are incoherent, in advance to get a fair description of the role coherence plays, let us consider the BB84 protocol~\cite{BB.84}. 

The goal of  BB84 is to distribute random keys in a safe way. This is achieved by sending  qubits through a quantum channel where each of them is  encoded by Alice in one of two possible non-commuting bases depending on the outcome of a random measurement. On the other side, Bob measures the qubit in a random basis. After repeating this many times, Alice and Bob publicly compare the bases chosen and keep only the data of the measurements that were done in the same bases. They can then compare a fraction of the remaining data to be certain (in the asymptotic limit) that nobody interfered. 
A naive approach to describe this protocol is that, as on each round the state of the qubit sent is diagonal in some basis, then the full protocol is classical. Why does this algorithm beat classical key distribution? The crucial departure from classical physics of this protocol is the random choice of bases and the fact that the security of the algorithm relies on the fact that the choice of basis in which the qubits are diagonal in is unknown to a possible eavesdropper Eve. Therefore, to analyse the security of this protocol, we need to take the perspective of Eve.
For her, whatever basis she assume to be incoherent, there will be some states sent by Alice that will be coherent, for a long enough sequence. That is because she does not know the outcomes of Alice measurements that defined the encoding bases. For this reason a fair description of what Eve can do needs to assume that she chooses the basis before the protocol starts.

As noted in the introduction, we only need one quantum system $Q$ together with a wire $W$ to make sense of this theory. We start by describing this case in some detail, before coming back to the case with multiple parties $Q=\otimes_i Q^i$ and wires $W=\otimes_j W^j$. For simplicity, to change the phases of the basis states is also assumed to be a free operation -- but as shown below this does not significantly alter the theory. We assume that both $Q$ and $W$ are finite dimensional systems, but don't keep their dimensions fixed (see footnote~\footnote{\label{fn:rep} For any ancilla $W_2$ to $W_1$ we fix an incoherent basis and label the full basis of the new $W$ by natural numbers and use implicitly the change of representation $|i(j,k)>_W \equiv |j>_{W_1} \otimes |k> _{W_2}$} for more details). The free operations in this case consist of iteratively applying the following elemental Kraus operations (which should form a CPTP map on $W\otimes Q$, that is $\sum_{\alpha} {K^{\alpha}}\mdag  K^{\alpha} = \id$) and post-selecting (for a neater notation we only write out the spaces on which the Kraus operations act non-trivially, on all others the identity operation is assumed):

\begin{enumerate}
	\item \emph{Permutations:} Permutations $\sigma$ of basis states in the basis ${\cal Z}$ on the wire W, $\sum_i\ketbra{\sigma(i)}{i}_{W}$.
	\item \emph{Phases:} Diagonal unitary evolutions in the basis ${\cal Z}$ on $W$, $\sum_j e^{\im \phi(j)} \ketbra{j}{j}_{W}$.
	\item \emph{Observed quantum operations:} Any generalized measurements on $Q$, encoding its outcomes as an incoherent state of an ancillary subsystem $W_a$ of $W$, $ \ket{\alpha}_{W_a}\otimes F^{\alpha}_{Q}$, for Kraus operators $F^{\alpha}_{Q}$ acting on $Q$.
	\item \emph{Classical to quantum forwarding:} Transfer a subsystem 
	$W_s$ of $W$ to a subsystem $Q_t$ of $Q$:   $\sum_{j} \bra{j}_{W_s}\otimes\ket{j}_{Q_t}$.
\end{enumerate} 
For simplicity, here we assume the incoherent basis $\cal Z$ to be the same for all times, but one can easily get the more general setting by replacing the free operations $K^\alpha$ by $(U_W(t,t_0) \otimes \id_Q) K^\alpha$, where $U_W(t,t_0)$ defines the (necessarily predefined, see above) change of basis for the time-step $t_0\rightarrow t$. Note that this includes replacing the identity by $(U_W(t,t_0) \otimes \id_Q)$ and that the identity operation itself is not a free operation any more. Sticking to this rule the results stay unchanged, because concatenating free operations on $t_0, t_1$ and $t_1,t_2$ yields free operations on $t_0, t_2$. This is relevant if one wants to change the basis one calls classical in time, as is usually the case in settings where discord is thought to be a meaningful measure. We will come back to this 
in subsection~\ref{subsec:quantum_discord}.

Note also that we treat encoding and decoding asymmetrically. The rational we employ here is that  it is hard to encode quantum information in a wire, but if coherence is provided and it is possible to sustain, transport and control it, you may very well also be able to use it. That is why we only allow to encode classical information in the wire, while any state can be retrieved.
We call the set of free operations in this theory of local operations and physical wires $LOP(W \IC Q)=LOP$, were here we use the symbols $\leftarrow$ and $\rightarrow$ for classical encoding, while $\leftsquigarrow$ and $\rightsquigarrow$ are used for transferring a quantum system. Items 3. and 4. are depicted in Fig.~\ref{fig:operations}.
\begin{figure}[h]
	\centering
	\includegraphics[width=9cm]{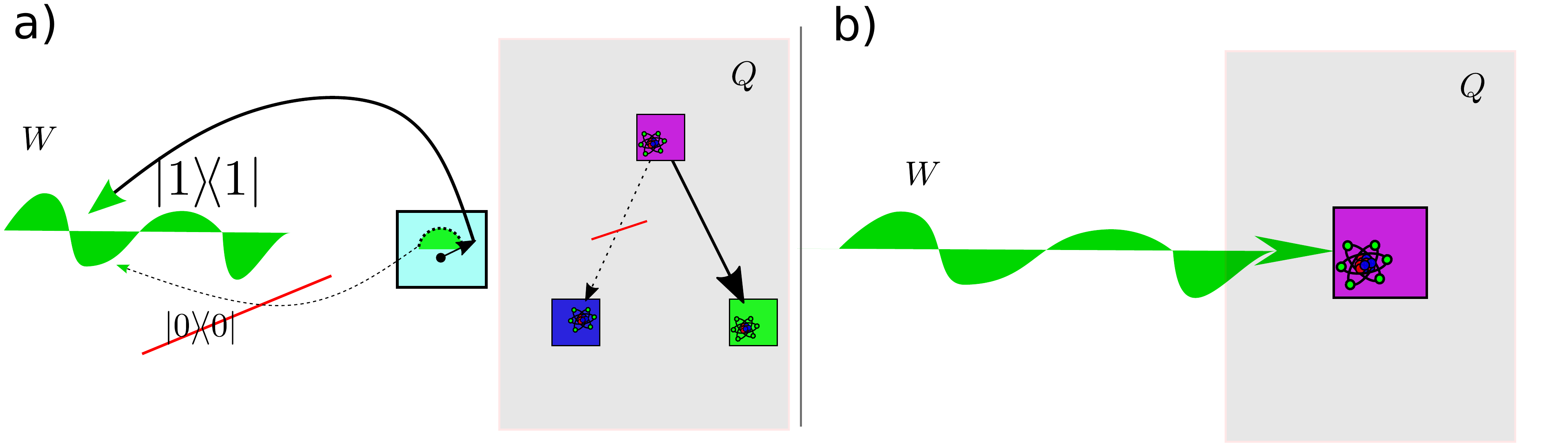}
	\caption{The picture depicts items 3. and 4. of the elemental operations defining LOP. Item 3., depicted in sub-figure a), means that to do operations on the quantum state, measure what happens and use the incoherent basis of the wire to encode the observation is considered to be a  classical operation on the wire and therefore free. Item 4, depicted in sub-figure b), means that the wire is effectively treated as a quantum system. This means that it can interact with the quantum system in a quantum way. However if the wire is an extended object it might be hard to change its state in a controlled way by this interaction. Therefore it is considered a free operation to forward the state of a subsystem of the wire to the quantum side, but not to alter the state of the wire in a quantum way.}
	\label{fig:operations}
\end{figure}
As promised, we first show that the phases are not really relevant (remember that all proofs can be found in the appendix):
\begin{prop}\label{prop:phases}
	Any operation in LOP can be done with arbitrarily high probability of success by a combination of permutations, observed quantum operations and classical to quantum forwarding.
\end{prop}
It is easy to see that the set of free states that can be prepared (probabilistically or not) by these operations is given by the so-called incoherent-quantum states, $CQ^{(n)}_{{\cal Z}}=\{\rho\mid\rho=\sum_{m}p_m \ketbra {m}{m}_W \otimes {\sigma_m}_Q, |m\rangle \in {\cal Z}  \}$~\cite{Streltsov2015b,Ma2016}. 
Therefore LOP is a subset of the incoherent-quantum operations (IQO), that is given by any operations that map incoherent-quantum states to incoherent-quantum states, even after post-selection~\cite{Ma2016}.
It is not so easy to see how strict the inclusion is. We will discuss this question later in Prop.~\ref{prop:gennoninj}, once we have gathered more insight into the theory.

To get a clearer idea of the operations at hand, we will need to find a more compact form
of the operations belonging to LOP.
To this end and of independent interest it is helpful to know bijections 
that are elements of
LOP. Knowing bijections in a resource theory is useful because they define equivalence classes on states, in that all states that can be reached by bijections can be freely interconverted,
making them equivalent resources (as for any task one state can be used for, any state that can be mapped into it is at least as useful). A trivial bijection in $LOP$ is given by any unitary on $Q$, which means that any measure of the theory necessarily has to be invariant under basis changes of $Q$. The following lemma establishes such a bijection between the states on the wire $W$  and maximally correlated states on $W\otimes Q$ (also see~\cite{Streltsov2015b,Hu2016}).
\begin{lem}\label{lem:bij}	
	The operator $B:W\rightarrow W\otimes Q$, $B=\sum_i\ketbra{i}{i}_W\otimes \ket{i}_{Q}$ defines an injection and defining ${\cal B}[\rho]=(B\rho B\mdag)$,  ${\cal B}\in LOP$. Conversely there exists a map ${\cal B}^{-1}:W\otimes Q \rightarrow W$, with  ${\cal B}^{-1} \circ {\cal B} = \id_W$ and ${\cal B}^{-1} \in LOP$.
\end{lem}

We are now ready to give an equivalent characterization of the operations in $LOP$, which has the advantage of being an explicitly finite concatenation of maps of one fixed form. On \ one hand this might help to find a minimal number of necessary Kraus operators~\cite{Streltsov2017a}, on the other, a priori it is not obvious that such a simplification exists, since for LOCC, which is a very similar theory, the number of rounds needed for a transformation is unbounded~\cite{Chitambar2011}. The proposition states
that one can write any element of $LOP$ as a concatenation of $N$
$LOP$ maps (and $N$ is bounded by the Hilbert space dimension of the wire) that are composed of Kraus operators $K_j^{\alpha_{j}}$ having a specific functional form. Due to the possibility of post-selection, one has to consider different paths given by the outcomes of the generalized measurements, that is, which Kraus operators $K_j^{\alpha_{j}}$ have been measured; after $t$ maps (labelled from $1$ to $t$), these paths are denoted by $\vec{\alpha}_{t}$, and $\alpha_{t+1}$ denotes the possible outcomes of the map labelled by $t+1$. Both the maps and the length of the protocol may vary depending on the outcomes of previous measurements, but for any path $\vec{\alpha}$, the total length $N=N(\vec{\alpha})$ can be restricted to be less or equal to the dimension of the Hilbert space of the wire. Fig.~\ref{fig:tree3} depicts a generic example for a protocol which acts on an initially three-dimensional wire and a quantum system.

\begin{figure}[h]
	\centering
	\includegraphics[width=9cm]{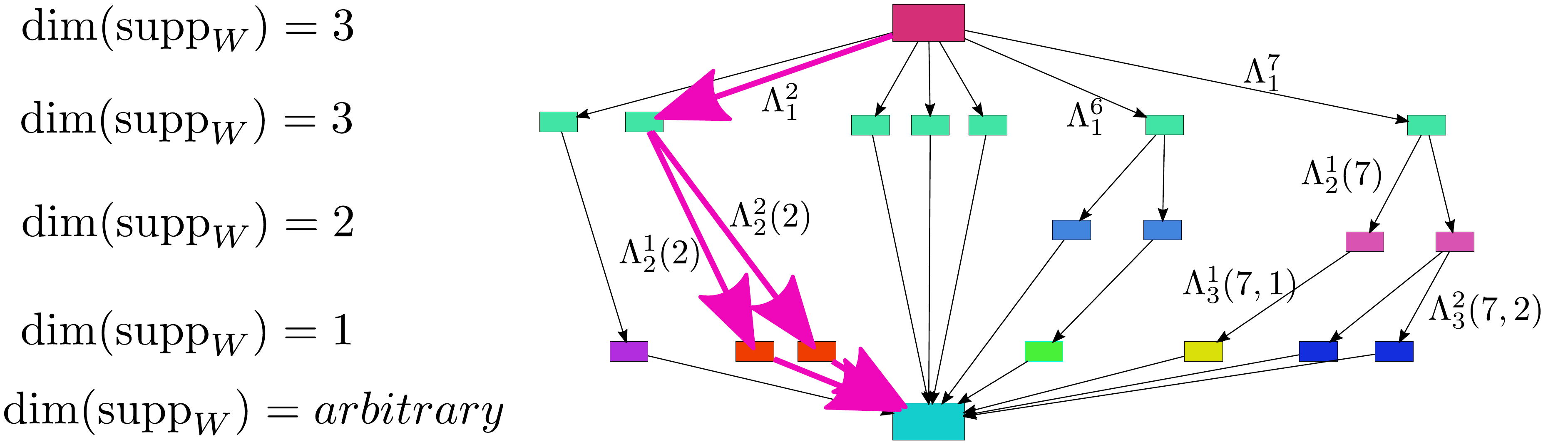}
	\caption{Depicted is a diagram of a possible protocol defining an LOP map as described in Prop.~\ref{prop:form} for the special case of an input state that is three-dimensional on the wire. All the arrows starting from the same node represent a CPTP map. Each arrow represents a sub-map with given outcome (labelled from 1 on the left to $n$ on the right). The dimension of the maximal support of the input or output states on the wire is stated on the left of the diagram. The branches can be grouped into four different types, where the types differ by how the dimension of the maximal support of the state on the wire (given wlog by $r(\vec{\alpha}_{t-1})$, after applying $t$ maps with the outcomes $\vec{\alpha}_{t}$), changes during the protocol. Any possible branch type is depicted at least once. In this picture the respective groups of outcomes for the four branch types are (we only name the outcome vectors up to the point that the group is defined): $\{(1),(2)\}$, $\{(3),(4),(5)\}$, $\{(6,1)\}$ and $\{(6,2),(7)\}$. See the Appendix for more details on the branch types.  The highlighted branch is given by $\Lambda_2(2)\circ \Lambda_1^2$, that is, the CPTP map $\Lambda_2(2)=\Lambda_2^1(2)+\Lambda_2^2(2)$ (composed by two sub-maps), after getting the outcome 2 in the first map. The length of the protocol $N$ is 2 for these outcomes, while it is 3 for e.g. the outcomes (7,2). In principle the dimension of the support of the outcome state on the wire after applying an LOP map is arbitrary (though for the depicted protocol it is at most $3\times3+1\times2+7\times1=18$).}
	\label{fig:tree3}
\end{figure}

\begin{prop}\label{prop:form}
	Let $\Lambda$ be a CPTP map acting on $W\otimes Q$. $\Lambda$ is an LOP operation exactly if it can be written as a finite sequence of maps $\Lambda=\sum_{\vec{\alpha}_{N}}\Lambda_{N}^{\alpha_N}(\vec{\alpha}_{N-1}) \circ\ldots \circ\Lambda_1^{\alpha_1}$, (where $N \leq \dim(W)$ is the length of the protocol, possibly depending on the previous outcomes $\vec{\alpha}_{t}=(\alpha_1,\ldots,\alpha_{t})$, 	$\vec{\alpha}_{0}=\vec{\alpha}_{-1}=0$) with CPTP maps $\Lambda_t(\vec{\alpha}_{t-1})=\sum_{\alpha_{t}}\Lambda_t^{\alpha_t}(\vec{\alpha}_{t-1})$ 
	having a Kraus decomposition $ \Lambda_t(\vec{\alpha}_{t-1})[\rho]=\sum_{\alpha_{t}} K_t^{\alpha_{t}}(\vec{\alpha}_{t-1})\mdag \rho K_t^{\alpha_{t}}(\vec{\alpha}_{t-1})$ of the form
	\begin{align}\label{eq:form}
	K_t^{\alpha_{t}}(\vec{\alpha}_{t-1})
	=
	\sum_{i=1}^{r(\vec{\alpha}_{t-2})}  \ketbra{\sigma_{\vec{\alpha}_{t}} \circ \min[r(\vec{\alpha}_{t-1}),i]}{i}_W
	\otimes E_t^{ \alpha_{t}}\left(\vec{\alpha}_{t-1},i\right),
	\end{align}
	with $\sigma_{\vec{\alpha}_{t}}$ being an injective map to the positive labels of a new incoherent basis (see footnote~\cite{Note1}) and $E_t^{ \alpha_{t}}\left(\vec{\alpha}_{t-1},i\right)$ being an arbitrary operator acting on $Q$, potentially depending on previous outcomes and controlled by the populations of the wire and $1\leq r(\vec{\alpha}_{t})<r(\vec{\alpha}_{t-1})$ for any $t\in \{1,\ldots, N-1\}$ and $r(0)=\dim(W)$.
\end{prop}

 The above proposition is useful to connect other theories that have been discussed in the literature with the one presented here. From Prop.~\ref{prop:form} it follows that destructive measurements on any subsystem $W_s$ of $W$ are free (destructive measurements being a set of Kraus operators mapping to a one-dimensional Hilbert space, which consists of only one, trivially incoherent, state). This can also be more directly seen, as one can forward the subsystem $W_s$ from the wire to the quantum side and perform the measurement there. Note that one can understand any Positive Operator Valued Measure (POVM) as a destructive measurement, since one is not interested in the outcome \cite{NC.00}, which can be stated as: \emph{any POVM  can be implemented inside LOP}.  
Also note that any special incoherent operation SIO can be performed on $W$. Special incoherent operations have Kraus operators which commute with dephasing, forcing them to have the form $F^{\alpha}=\sum_i c_{\alpha}(i)\ketbra{\sigma_{\alpha}(i)}{i}$, for permutations $\sigma_{\alpha}$ and complex $c_{\alpha}(i)$~\cite{Winter2016,Yadin2016}. Even more restrictive, physical incoherent operations PIO are special incoherent operations where the permutations are fixed for all the Kraus operators ($\sigma_{\alpha}=\sigma$). Both are obviously special cases of the form given in Prop.~\ref{prop:form}. Indeed,
\begin{prop}[SIO, PIO and LOP]\label{prop:PIO}
	Let $\Lambda$ be a CPTP map acting on $W\otimes Q$. $\Lambda$ is an LOP operation exactly if it can be written as a sequence of maps $\Lambda=\Lambda_M(\vec{\alpha}_{M-1}) \circ\ldots \circ\Lambda_1$, for some finite $M$,
	where each $\Lambda_i$ is
	\begin{itemize}
		\item[(a)]  a physical incoherent operation on $W$
		\item[or (b)] a destructive measurement in one fully coherent basis of a subsystem of $W$
		\item[or (c)]	a controlled unitary (${\bf U}_{\rm control}=\sum_m \ketbra{m}{m}_{\rm control} \otimes {\bf U}_{{\rm target}}(m)$) with control $W$ and target $Q$ 
		\item[or (d)] a generalized measurement of $Q$, encoding the result on $W$ ($\rho \mapsto \sum \limits_{\alpha} \ketbra{\alpha}{\alpha}_W\otimes K^{\alpha}\rho {K^{\alpha}}\mdag$).
	\end{itemize}
	One can equivalently replace item (a) by ``a special incoherent operation on $I$''.
\end{prop}
Noting that items (a), (c) and (d) together form an effective theory of the special incoherent operations
by restricting the theory to the effect on the wire~\cite{Yadin2016}, we find that the present approach (apart from giving a completely different motivation) only differs by allowing measurements. This difference however, is crucial; that special incoherent operations commute with dephasing means that apart from being unable to create coherences, they affect populations only depending on populations, making them incoherent in a very strict sense. So strict in fact that coherences don't have any effect that can be measured by free operations, meaning that coherences are useless for any observable task in that resource theory and there are thus no resource states in the theory. In contrast, ideally (but maybe not always necessary) one would find that having enough non-free states at hand removes completely any restrictions of the theory, showing that this is really the resource that helps overcome the restriction and giving sense to questions of how much resources one needs for a given task. Exactly this we find in the present approach: if supplemented by enough coherent ancillary states, $LOP$ can be used to achieve any desired quantum operation.
\begin{thm}[Resource states]\label{prop:resource}
	Let $\Lambda$ be a CPTP map acting on $W_1\otimes Q$, with $W_1$ having dimension $d$. Let $\ket{\psi} = \sum_{i=1}^d \frac{1}{\sqrt{d}} \ket{i}$ be a maximally coherent state on $W_2$. Then there is an operation $\Lambda'\in LOP(W_2\otimes W_1\IC Q)$, with $\tr_{W_2} [\Lambda'[\ketbra{\psi}{\psi}_{W_2}\otimes\rho_{W_1,Q}]]=\Lambda[\rho_{W_1,Q}]$.
\end{thm}
We conclude that coherence is a meaningful resource in $LOP$.
Furthermore, from Lem.~\ref{lem:bij} it follows that the resource of coherence $\sum_{i=1}^d \frac{1}{\sqrt{d}} \ket{i}_W$ can be reversibly converted into entanglement between the wire and the system $Q$, $\sum_{i=1}^d \frac{1}{\sqrt{d}} \ket{ii}_{WQ}$ which hence is  an equivalent resource of the theory (also see~\cite{Streltsov2015b,MEKP.16}). 

The above properties resemble much those found for incoherent-quantum operations (IQO)~\cite{Ma2016} and as stated above, the free operations of this theory contain $LOP$. But how strict is this inclusion?
In the following proposition, we show that in general incoherent-quantum operations cannot be performed by LOP operations. However,
\begin{prop}[LOP and IQO]\label{prop:gennoninj}
	Be $\Lambda$ an incoherent-quantum operation on $W_1\otimes Q$, which is exactly the case if it is CPTP and has a Kraus decomposition with Kraus operators of the form $K^{\alpha}=\sum_i \ketbra{f_{\alpha}(i)}{i}_{W_1} \otimes E^{\alpha}(i)$, for some functions $f_{\alpha}$ acting on the labels of the incoherent basis and some operators $E^{\alpha}(i)$ acting on $Q$. Let $d := \dim(W_1)$. If $d=2$, $\Lambda \in LOP$. For $d\geq 3$ $LOP\neq IQO$, but there is a stochastic implementation of the map in $LOP$ with 
	a success rate of at least
	$1/d$; i.e. there is an operation $\Lambda' \in LOP(W_2\otimes W_1\IC Q)$  with $\Lambda'[\ketbra{0}{0}_{W_2}\otimes\rho]=\ketbra{0}{0}_{W_2}\otimes\Lambda_0'[\rho]+\ketbra{1}{1}_{W_2}\otimes\Lambda_1'[\rho]$ with $\Lambda_1'[\rho]=\Lambda[ \rho]/d\; \forall \rho$.
\end{prop}

Meaning that even though the theories $LOP$ and $IQO$ differ, they are very similar. 
Equally close are the theories they induce on the wire (by tracing out $Q$ at the end). We conclude that the effective theory we 
obtain by restricting $LOP$ to the wire -- even though being operationally motivated -- is in many aspects similar to the abstractly motivated theory of coherence introduced in~\cite{BCP.13}, were the free operations are any that can be decomposed into Kraus operators that cannot create coherence. Specifically, the free states are the same, the resource states are the same and the same amount of resources is needed to remove all restrictions of the theories, in both theories any destructive measurements can be performed, for qubits the theories are equivalent and for higher dimensions the theories are stochastically equivalent. 
Furthermore, for a given $\epsilon$, having enough copies of any state which is not incoherent-quantum, LOP allows to prepare a maximally coherent state with fidelity $f>1-\epsilon$ and probability $p> 1-\epsilon$ (see Appendix~\ref{sec:resources}).
Therefore, since maximally coherent states can be used as a resource to implement any quantum operation (Thm.~\ref{prop:resource}), all not incoherent-quantum states are resources for non-classicality of the wire in the present setting. This is true independently of how one defines non-classicality, as long as it does not include full quantum theory. More specifically, since for any free operation the state remains incoherent-quantum if it was at the beginning, one can make a stochastic model for any free operation, which correctly describes the change of the reduce state of the wire, as long as at the beginning the full state is incoherent quantum. Explicitly, one just needs to replace the elemental operations, that is, permutations, tracing out and encoding a classical outcome of a measurement by their obvious non-contextual (and local) classical counterparts. For the measurement this means encoding the classical probabilities as defined by the Born rule and the current state of the external quantum system to measure, but this can only depend on some previous populations of the wire: its previous classical states. Even so, having enough copies of any state which is not incoherent-quantum will allow to do any operation on the full system, including those that have no non-contextual (or non-local, see the next section about entanglement) classical counterparts.

\subsection{Quantum Discord}\label{subsec:quantum_discord}
Before continuing and applying the tools we have just developed to multipartite entanglement, let us stop a moment and draw the connection to quantum discord. In a sense to be made precise now, $LOP$ can be seen as a theory of quantum discord. As the free states of $LOP$ are the incoherent-quantum states defined above, the theory can be seen as quantifying how much a state differs from an incoherent-quantum state, that is, by definition, how much basis-dependent quantum discord it has (called measurement dependent in~\cite{Modi2012}). Take any measure of basis-dependent discord which can at the same time be normalized such that it yields $1$ for a singlet state. Assume also that the measure  is upper bounded by the entropy of the local state of the wire. Then, by minimizing this quantity over all possible bases we obtain a measure of discord satisfying the properties stated in section II.A.1. of the review~\cite{Modi2012}. While starting from coherence theory it seems strange to do this optimization, since one starts by choosing a natural basis as incoherent and only then the notion of coherence makes sense, the present approach is much nearer to one original setting in which quantum discord was introduced~\cite{OZ.01}. Indeed the theory of coherence of the wire in the present approach is an effective theory and the incoherent basis is arbitrary, it just depends on what one assumes to be the default basis to encode the information in. 
This allows to interpret quantum discord as a natural lower bound that defines the non-classical resource one has to assume at least, independently of the basis chosen. However, this is not the same as to say that quantum discord is a resource of non-classicality. What happens is, that if one has two quantum states that have zero discord, there still might be no basis for which both have zero basis dependent discord simultaneously and it is thus not surprising that their mixture might not be discord-free (also see the discussion in~\cite{OZ.01} and Fig.~2 in~\cite{MEKP.16}). This means that in the present framework we reproduce the interpretation of quantum discord as how non-classical one has to assume the state of the wire at least, when seen as a part of the full system $W\otimes Q$ (similarly as was the motivation for its introduction in~\cite{OZ.01}) -- if one considers only one state (even at different times, since, as noted above, one can change the basis in time). But having more than one possible state (for instance if one does not consider the evolution of one given state, but a protocol defined for any input state, such as in quantum cryptography), it is necessary to fix one basis (possibly a different one at different times) for all the states one considers, to have a meaningful quantity. 
It is in this sense that the present approach defines a resource theory of quantum discord, even though discord is not (and should not be!) a monotone and thus a measure of the resource theory, as discord is not a resource, it is an indicator of non-classicality: if a state has non-zero discord, this means that for whatever basis one chooses as incoherent, the state is a resource for non-classicality with the present framework.

\section{Coherence cost of entanglement}
The aim of this section is to better understand multipartite entanglement by looking at the coherence needed to generate it. This is a natural approach as entanglement is always a result of coherent interactions happened in the past. 
In the case of bipartite entanglement the theory in our setting is the one depicted in Fig.~\ref{fig1}, which we denote by $LOP({Q}^1 \CI {W}\IC {Q}^2)$ and consists of two quantum systems ${Q}^1$ and ${Q}^2$ connected by a wire ${W}$. The elemental free operations are the free operations of $LOP({W}\IC {Q}^1)$ together with the ones in $LOP({W}\IC {Q}^2)$, meaning that any operation is given by the composition of these operations, possibly with post-selection. In general, we will call $LOP({Q}^{1} \CI {W}^{1} \IC {Q}^{2} \CI {W}^{2 }\IC \ldots)$ the set of operations consisting of concatenating operations in the corresponding sets  $LOP({W}^{j} \IC {Q}^{j\;(or\;j+1)})$. The notation is explained in Fig.~\ref{fig:notation} The minimal amount of coherence needed to create pure bipartite entanglement directly follows from Lem.~\ref{lem:bij} (also see~\cite{Streltsov2015b}). One can produce a pure state on ${Q}^1\otimes {Q}^2$ which in its Schmidt-decomposition~\cite{NC.00} is given by $\ket{\psi}_{{Q}^1,{Q}^2}=\sum_i c_i \ket{ii}_{{Q}^1,{Q}^2}$ from $\sum_i c_i \ket{i}_{W}$ on ${W}$ by applying Lem.~\ref{lem:bij} to get $\ket{\psi}_{{Q}^1,{W}}$ and then using classical to quantum forwarding on ${W}\otimes {Q}^2$ to get the wanted state. On the other hand one also sees that the production is optimal since one needs to have at least that amount of entanglement on the bipartite cut ${Q}^1 \mid {W}\otimes {Q}^2$ (remember that entanglement and coherence are equivalent resources for $LOP({Q}^{1} \CI {W})$). As shown below, the connection between coherence and entanglement in the bipartite case is even stronger; a maximally correlated state can be used as a resource for a $LOCC$ transformation exactly if the equivalent coherent state can be used as a resource to do the transformation under $LOP$. 
This puts the connection between coherence and entanglement that has recently attracted a lot of attention on the level of resources and operations instead of measures~\cite{Streltsov2015b,Yao2015,Chitambar2016a,Chitambar2016c,Streltsov2016a,Hu2016,Hu2016a,Hu2017,Ma2016,Yadin2016,Zhu2017,Radhakrishnan2016,Streltsov2017,Girolami2017,Bu2017}. It will be useful to introduce the notation ${W}^{j}=W^{j}_1\otimes  W^{j}_2\ldots$ and ${Q}^{j}=Q^{j}_1\otimes Q^{j}_2\ldots$, with the convention that the upper index labels the local systems and the lower their respective subsystems, whose number may vary. The full systems are referred to by ${Q}=\otimes_j {Q}^{j}$, ${W}=\otimes_j {W}^{j}$. 
\begin{figure}[h]
	\centering
	\includegraphics[width=9cm]{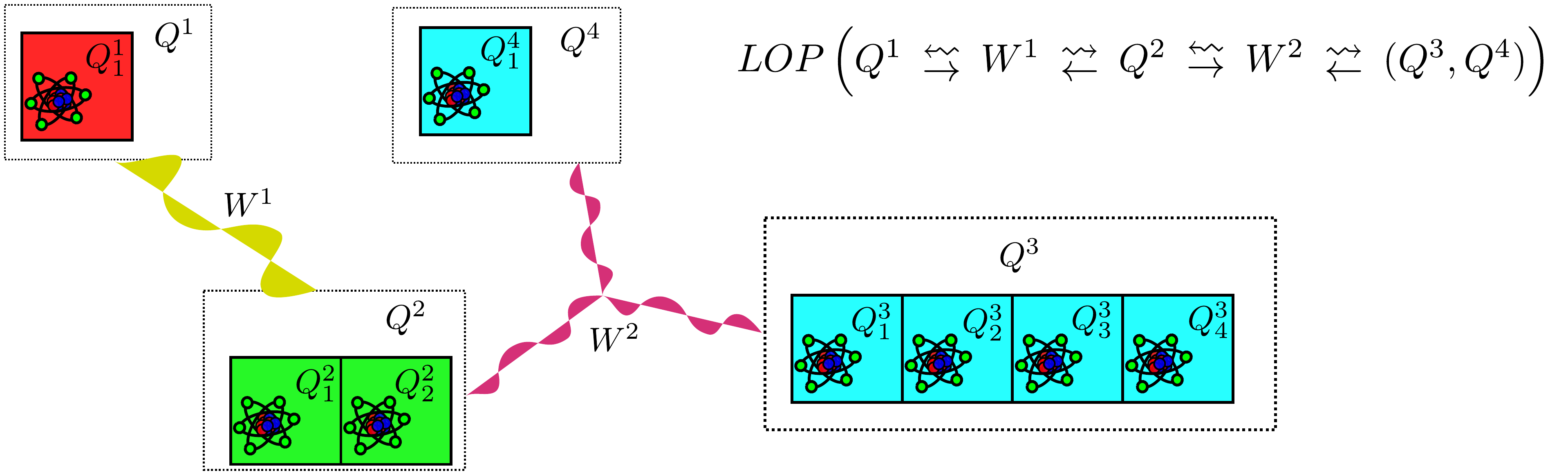}
	\caption{The figure shows a possible wiring for four parties. The first wire connects the systems $Q^1$ and $Q^2$ the second connects $Q^2$, $Q^3$ and $Q^4$, resulting in the theory $LOP(Q^1\CI W^1 \IC Q^2 \CI W^2\IC (Q^3,Q^4))$. Each party (and also the wires) might consist of more than one quantum system and changing the number of subsystems of each party is a free operation in LOP.}
	\label{fig:notation}
\end{figure}

\begin{thm}\label{thm:bipartite}
		Let $\eta_{LOCC}=\sum_{ij} r_{i,j} \ketbra{i}{j}_{Q^1_1}\otimes\ketbra{i}{j}_{Q^2_1}$ be a maximally correlated state (in arbitrary orthonormal local bases of $Q^1_1\otimes Q^2_1$) and $\eta_{LOP}=\sum_{ij} r_{i,j} \ketbra{i}{j}_{{W}}$ be a corresponding state in the incoherent basis $\cal Z$ of ${W}$.
		If $\Lambda$ is  a CPTP map on ${Q}^1\otimes {Q}^2$, then the following statements are equivalent:
			\begin{align*}
			1.\;\;& \exists 	\Lambda_{LOCC}\in LOCC({Q}^1,{Q}^2): \\
			&\Lambda_{LOCC} \left[ \eta_{LOCC} \otimes \rho_{Q^1_2,Q^2_2} \right] 	=\Lambda[\rho_{Q^1_2,Q^2_2}] \; \forall \rho_{Q^1_2,Q^2_2}\\
			2.\;\;& \exists \Lambda_{LOP}\in LOP({Q}^1\CI {W} \IC {Q}^2): 
			\\&  \Lambda_{LOP} \left[ \eta_{LOP} \otimes \rho_{Q^1_2,Q^2_2} \right] 	=\Lambda[\rho_{Q^1_2,Q^2_2}] \; \forall \rho_{Q^1_2,Q^2_2}.
			\end{align*}
	\end{thm}
	We now use the common definition that $\rho \stackrel{\cal O}{\rightarrow} \sigma$ means that there is a map $\Lambda \in \cal O$, with $\Lambda[\rho]=\sigma$ (for some space of superoperators $\cal O$). The following corollary then follows directly by taking $\Lambda$ in the theorem, to be the preparation of the state $\sigma_{{Q}^1,{Q}^2}$:
\begin{cor}
	Let $\eta_{LOCC}=\sum_{ij} r_{i,j} \ketbra{i}{j}_{Q^1_1}\otimes\ketbra{i}{j}_{Q^2_1}$ be a maximally correlated state (in arbitrary orthonormal local bases) and $\eta_{LOP}=\sum_{ij} r_{i,j} \ketbra{i}{j}_{{W}}$ be a corresponding state in the incoherent basis $\cal Z$ of ${W}$. Then:
	\begin{align*}
	&		\eta_{LOCC} \otimes \rho_{Q^1_2,Q^2_2} \stackrel{LOCC({Q}^1,{Q}^2)}{\rightarrow} \sigma_{{Q}^1,{Q}^2} \\
	\Leftrightarrow \,& \eta_{LOP} \otimes \rho_{Q^1_2,Q^2_2} \stackrel{LOP({Q}^1\CI {W} \IC {Q}^2)}{\rightarrow} \sigma_{{Q}^1,{Q}^2}.
	\end{align*}
	
\end{cor}

Having shown this very strong connection between coherence and bipartite entanglement, we move to the multipartite $LOCC$ case. While in the bipartite case it is quite clear how one needs to explicitly implement the wires (as there is one clearly simplest way to connect two parties), in the multipartite case the situation is more complex. To be able to do any operation by using enough coherence, each party needs to be connected to all others (possibly indirectly over third parties). On the other hand one also does not need to have two parties connected in two different ways and avoiding to have double connections simplifies the theories.
We use the short-hand notation $LOP( {W}^{1} \IC ({Q}^{1},{Q}^2,\ldots))$ for the case of one wire connecting different quantum systems (i.e. concatenating operations in the sets  $LOP({W}^{1} \IC {Q}^{i})$). For the example of three parties ${Q}=A\otimes B\otimes C$, we are left with the 3+1 theories of the two types depicted in figure~\ref{fig:fig3}; namely the three ways ($A$, $B$ or $C$) of choosing ${Q}^2$ in $LOP({Q}^1\CI {W}^1 \IC {Q}^2\CI {W}^2 \IC {Q}^3)$ and $LOP({W} \IC (A,B,C))$.
\begin{figure}
	\centering
	\label{fig:fig3}
	\includegraphics[width=9cm]{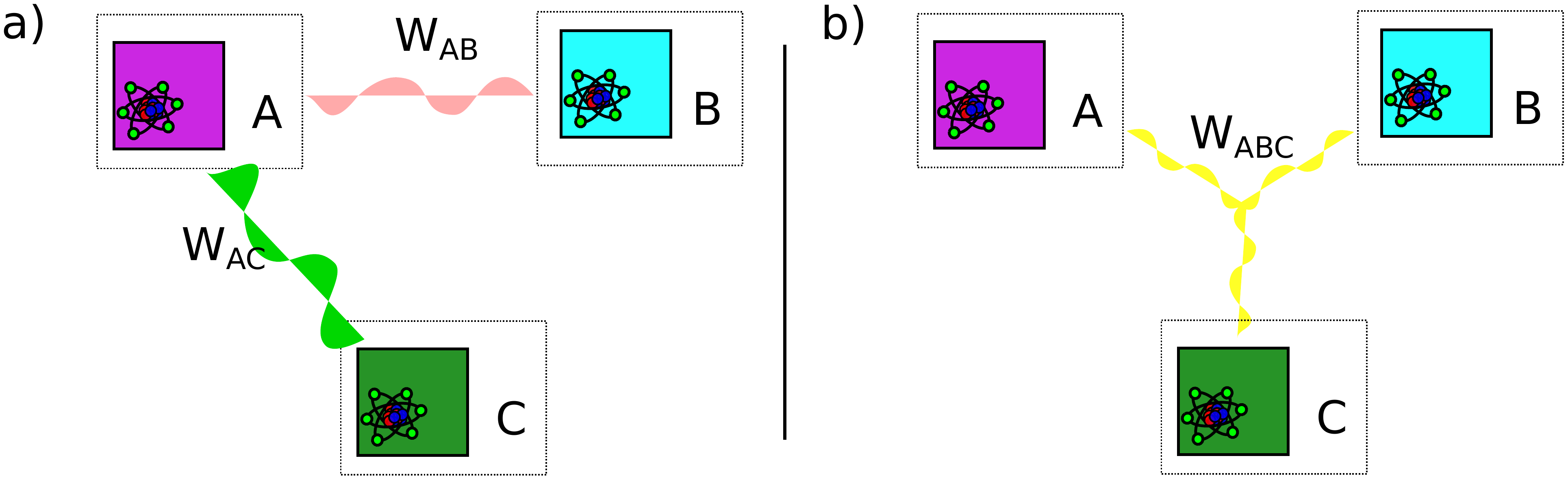}
	\caption{Two wiring schemes that are equivalent in the LOCC paradigm, but inequivalent according to LOP. For LOCC, i.e. for incoherent states of the wiring, every operation can be performed by acting over the extremes of each wire, without a direct interaction among wires. On the other hand, if wires are quantum, direct interaction among wires enlarge the set of possible inequivalent operations. As a result, the initial coherence necessary to prepare a multipartite entangled state depends on the wiring topology.}
\end{figure}
In general we get that following the rules explained above, the possible generalization of $N$-partite LOCC are exactly given by the theories 
\begin{align*}
	LOP(({Q}^{\sigma(1)}\ldots {Q}^{\sigma(f_1)}) \CI {W}^{1} \IC ({Q}^{\sigma(f_1 +1)}\ldots {Q}^{\sigma(f_2)}) \\
	 \ldots \CI {W}^{L}\IC({Q}^{\sigma(f_L +1)}\ldots {Q}^{\sigma(N)})),
\end{align*}
where $\sigma$ is a permutation of the indexes denoting the local quantum systems. Of course one could also look at the union of all these theories, which has the advantage of providing a completely unified view, while having the disadvantage of being excessively complicated. 
What all these theories have in common is that they are generalizations of multipartite LOCC on the quantum side, that reduce to it if one starts with an incoherent state on ${W}$, which is made precise in the following remark.
\begin{rem} \label{rem:generalization}
	$\forall \rho_{W}=\sum_i p_i \ketbra{i}{i}_{W}$,
	\begin{align*}
	1. \;&\tr_{W} (\Lambda \circ (\rho_{W}\otimes\id_Q))\in LOCC({Q}^1,\ldots,{Q}^n) \\
	&\forall \Lambda \in LOP({Q}^{i(1)} \CI {W}^{j(1)} \IC {Q}^{i(2)} \CI {W}^{j(2) }\IC \ldots),
	\\
	%also valid for GOIA
	2. \;&\Lambda\otimes\id_{W} \in LOP({Q}^{i(1)} \CI {W}^{j(1)} \IC {Q}^{i(2)} \CI {W}^{j(2) }\IC \ldots)\\
	&\forall \Lambda \in LOCC({Q}^1,\ldots,{Q}^n),
	\end{align*}
	with $i$ and $j$ denoting permutations on the index sets.
\end{rem}
In the bipartite case, this means that if the initial state of the wire is incoherent and the wire is in a product states with both parties, tracing out the wire at the end reduces the theory to LOCC. The reason is that the state of the wire can be copied and stored in a local register on both sides. Every operation above can then just be obtained by local operations and broadcasting of classical information, by updating the two local registers (just do the same permutations and phases on both registers, copy the information of measurements to both registers and trace out the corresponding part of a register on Bob's side if Alice's has been used in a quantum operation). We make this seemingly obvious statement precise, because intuition sometimes might be misleading; for instance it is possible to produce entanglement by sending a quantum particle back and forth that is never entangled with either of the two parties~\cite{Cubitt2003}. While in entanglement theory this is surprising, since one can entangle two parties without using entanglement to do so, with the present approach the same statement is very intuitive, as entanglement is simply not the only resource present, coherence and basis dependent discord are resources as well. In the present approach without having resources one cannot entangle two parties, while any amount of basis-dependent discord allows to do that (that is because, as noted above, one can distillate coherence from any resource and the statement then just follows from Thm.~\ref{thm:bipartite}). 
The above remark makes it easy to prove that the coherence needed to create entanglement gives a bound on the possible entanglement conversions, namely:
\begin{thm}\label{thm:conv}
	Let $\rho$, $\sigma$ be states on ${W}$. If $\exists \tau_{W}$, a state on ${W}$, s.t. $\tau_{W} \stackrel{LOP({Q}^{i(1)} \CI {W}^{j(1)} \ldots)}{\rightarrow} \rho$, but $\tau_{W} \stackrel{LOP({Q}^{i(1)} \CI {W}^{j(1)} \ldots)}{\nrightarrow} \sigma$, then it follows that
$\rho \stackrel{LOCC({Q}^1,\ldots,{Q}^n)}{\nrightarrow} \sigma$
\end{thm}
Note that for any of these theories (independently of the wiring) the free states are given by:
\begin{align*}
CQ^{(n)}_{{\cal Z}^{\otimes}}=\{\rho\mid\rho=\sum_{m}p_m \ketbra{m}{m}\otimes \rho_m, |m\rangle \in {\cal Z}^{\otimes}, \rho_m\in SEP(n)  \},
\end{align*}
with ${\cal Z}^{\otimes}$ the product basis of the incoherent basis of each wire, and $SEP(n)$ the set of $n-$partite separable states
$$
SEP(n)=\{\rho \mid \rho=\sum_{k} p_k \rho^{(k)}_1 \otimes\ldots\otimes \rho^{(k)}_n\}
$$
This implies that for any of these theories, 
\begin{equation}
\label{eq:ndelta}
R_{{\cal Z}^{\otimes},n}(\rho)=\min_{\sigma \in CQ_{{\cal Z}^{\otimes}}^{(n)}}S(\rho \|\sigma),
\end{equation}
with $S(\rho \|\sigma):=\tr(\rho (\log(\rho) -\log(\sigma)) )$ the relative entropy, 
defines an additive geometric monotone (see e.g.~\cite{MEKP.16}). In general, the evaluation of this quantity is an NP-Complete problem \cite{Huang.14}. However, due to the monotonicity  of the relative entropy, we notice that
\begin{equation}
\label{eq:ndelta2}
R_{{\cal Z}^{\otimes},n}(\rho)\geq \max(R_{{\cal Z}^{\otimes}}({\rm Tr}_{{Q}}\rho) , R_{n}^{\cal E}({\rm Tr}_{{W}}\rho))
\end{equation}
being $R_{{\cal Z}^{\otimes}}({\rm Tr}_{{Q}}\rho)$ the relative entropy of coherence on the wires, and $R_{n}^{\cal E}({\rm Tr}_{{W}}\rho)$
the relative entropy of entanglement \cite{VPRK.97,VP.01,Wei.08}. Due to the additivity of the relative entropy, if either $R_{{\cal Z}^{\otimes}}({\rm Tr}_{{Q}}\rho)=0$ or $R_{n}^{\cal E}({\rm Tr}_{{W}}\rho)=0$, Eq. (\ref{eq:ndelta}) turns into an equality, providing  a lower bound to the amount of initial coherence required to prepare an entangled state among the parties of ${Q}$~\cite{Streltsov2015b,Ma2016}.

As an example of how this perspective can be applied in the multipartite case let's revisit the case of the $\ket{\cal W}$ and of the $\ket{GHZ}$ state~\cite{DVC.00} (here only the results are discussed, the protocols for the conversions can be found in the Appendix). 

The relative entropy of entanglement of $\ket{GHZ_n}=\frac{|0\rangle^{\otimes n}  - |1\rangle^{\otimes n}}{\sqrt{2}}$ is $1$~\cite{VP.01} and indeed one can prepare the $\ket{GHZ_n}$ state by $LOP(W \IC ({Q}^{1},{Q}^2,\ldots,{Q}^n) )$ from $\frac{\ket{0}  +\ket{1}}{\sqrt{2}}$. For $|{\cal W}_n\rangle=\frac{\sum_{k=0}^{n-1}|0\rangle^{\otimes k}|1\rangle |0\rangle^{\otimes n-k-1}}{\sqrt{n}} $ 
the relative entropy of entanglement is $(n-1)\log_2 (n/(n-1))>1 \forall n>2$~\cite{Wei.08}. It is then a simple corollary of Thm.~\ref{thm:conv} that $\ket{GHZ_n}\stackrel{LOCC({Q}^2,\ldots,{Q}^n)}{\nrightarrow} \ket{{\cal W}_n}$. 

The second thing to note is that on any bipartition the $\ket{GHZ_n}$ state is LOCC equivalent to $\ket{GHZ_2}$, while the $\ket{ {\cal W}_3}$ state on any bipartition is LOCC equivalent to $1/\sqrt{3} (\ket{00}+\sqrt{2}\ket{11})$. Indeed one can in all three possible two-wire settings $LOP({Q}^1\CI {W}^1 \IC {Q}^2\CI {W}^2 \IC {Q}^3)$ prepare $\ket{ {\cal W}_3}$ from $1/\sqrt{6} (\ket{0}+\sqrt{2}\ket{1})_{W^1}\otimes(\ket{0}+\ket{1})_{W^2}$, while the bipartite entanglement one can produce on the bipartition $Q^1,(Q^2\otimes Q^3)$ is not enough to prepare $\ket{GHZ_3}$: as on any bipartition one needs to prepare a fully entangled qubit and this is equivalent to a maximally coherent qubit, the state with minimal coherence to prepare $\ket{GHZ_3}$ is given by $1/2 (\ket{0}+\ket{1})_{W^1}\otimes(\ket{0}+\ket{1})_{W^2}$, which is strictly more coherent on $W_1$ than $1/\sqrt{6} (\ket{0}+\sqrt{2}\ket{1})_{W^1}\otimes(\ket{0}+\ket{1})_{W^2}$. Again as a corollary of Thm.~\ref{thm:conv}, we have that $\ket{{\cal W}_3}\stackrel{LOCC(Q^1,Q^2,Q^3)}{\nrightarrow} \ket{GHZ_3}$. A strong indication that the resources in the different types of wirings correspond to different types of entanglement.
\section{Conclusion}
Recently there has been considerable interest in the connection between coherence, discord and entanglement~\cite{Streltsov2015b,Yao2015,Chitambar2016a,Chitambar2016c,Streltsov2016a,Hu2016,Hu2016a,Hu2017,Ma2016,Yadin2016,Zhu2017,Radhakrishnan2016,Streltsov2017,Girolami2017,Bu2017}. But the connection was made on the level of quantifiers and measures. The current paper shows that by generalizing the fundamental theory of entanglement -- LOCC, one obtains a connection between coherence, discord and entanglement that is even deeper, namely on the level of the operations themselves. In this sense the current approach defines a theory that lies at the root of entanglement and it seems natural to assume that it will be useful to assess the interplay between different resources of quantumness in complex settings, as is needed for instance if one wants to quantify the resources needed for quantum  algorithms (see e.g.~\cite{Shahandeh2017}). Moreover, while we exemplified here that the present approach can give a clear explanation of the difference between some very basic forms of multipartite entanglement by exemplifying that different forms are optimal in different settings, it remains an interesting open question whether it yields such an explicative power also in more general cases and how it connects to known structures in entanglement theory, such as explained in e.g.~\cite{Walter2016,Lamata2007,Backens2017,Vicente2012,Regula2014}. 

Furthermore, there is an ongoing discussion over what is the ``right'' theory of (speakable) coherence~\cite{Chitambar2016b,Yadin2016,Marvian2016a,Streltsov2017,Vicente2017,Biswas2017}. While we do not claim to close this discussion (nor actually that it can be conclusively closed, as it always will depend on the setting one is interested in), we note that the effective resource theory of coherence emerging from our approach is to our knowledge the first which is built on operationally meaningful elemental operations while still having coherence as a resource in the sense of the name, meaning that enough ancillary coherence can completely lift the restrictions imposed by the resource theory. 

{\bf \em Acknowledgments ---}
The work of D.E., T.T. and M.B.P. is supported by the ERC Synergy Grant BioQ, the EU project QUCHIP and the US Army Research Office.  J. M. Matera is supported by CONICET.

\clearpage
\appendix
\section{Proofs on the structure of LOP}
We start with the proof of Lem~\ref{lem:bij}, introducing a bijection between the wire and maximally correlated states on the quantum system with the wire within $LOP$.
\begin{lem*}[\ref{lem:bij}]
		
	The operator $B:W\rightarrow W\otimes Q$, $B=\sum_i\ketbra{i}{i}_W\otimes \ket{i}_{Q}$ defines an injection and defining ${\cal B}[\rho]=(B\rho B\mdag)$,  ${\cal B}\in LOP$. Conversely there exists a map ${\cal B}^{-1}:W\otimes Q \rightarrow W$, with  ${\cal B}^{-1} \circ {\cal B} = \id_W$ and ${\cal B}^{-1} \in LOP$.
\end{lem*}
\begin{proof}
	The operator $B=\sum_i \ketbra{i}{i}_W\otimes\ket{i}_{Q}$ can be implemented by a sequence of maps that will be described in the following in terms of their Kraus operators. To this end we start with $W=W_1$ and $Q=1$ and apply 
	\begin{enumerate}
		\item $\ket{0}_{W_2} $
		\item $\sum_{i,j}\ketbra{i}{i}_{W_1}\otimes\ketbra{i+j}{j}_{W_2} $
		\item $\sum_{i,j}\bra{j}_{W_2}\otimes \ket{j}_{Q} $,
	\end{enumerate}
	where the first map is a single outcome measurement of $Q$ storing the outcome in an ancilla $W_2$ (observed quantum operation), the second is a permutation on $W$, and the third is the forwarding of the system $W_2$ to $Q$. Identifying $W=W_1$, we get the desired operation.
	
	For the converse, we first apply a measurement in the Fourier basis, followed by a correction of the phase on $W$.
	\begin{enumerate}
		\item	$FT^k=\bra{\hat{k}}_{Q}=\sum_{j=1}^{d}  \frac{e^{2 \pi \im k j /d} }{\sqrt{d}} \bra{j}_Q$
		\item  $D(k)= \sum_{i=1}^{d} \ketbra{i}{i}_W   e^{-2 \pi \im k i / d} $,
	\end{enumerate}
	resulting in the action ${B^{-1}}^k=\sum_{i=1}^{d} \sum_{j=1}^{d} \frac{e^{2 \pi \im k (j-i) / d} }{\sqrt{d}} \ketbra{i}{i}_W \otimes \bra{j}_Q$.
	Obviously the map defined by these Kraus operators is an element of $LOP$ and a left-inverse of ${\cal B}$, as required.
\end{proof}
Note that the basis on the quantum side can be chosen arbitrary in the above lemma. We continue with the proof of Prop~\ref{prop:phases}.

\begin{prop*}[\ref{prop:phases}]
	Any operation in LOP can be done with arbitrarily high probability of success by a combination of permutations, observed quantum operations and classical to quantum forwarding.
\end{prop*}
\begin{proof}
Let $W$ have dimension $d$. The only thing to show is that indeed one can change the phases on $W$ in the above framework. 
	 Let $U=\sum_j e^{\im \phi(j) } \ketbra{j}{j}_W$ be the wanted phase shift. We start by identifying $Q=Q_1$.
	 The protocol is:
	\begin{enumerate}
		\item $B=\sum_i \ketbra{i}{i}_W\otimes\ket{i}_{Q_2}$ as in Lem~\ref{lem:bij}
		\item $U_{Q_2}=\sum_j e^{\im \phi(j) } \ketbra{j}{j}_{Q_2}$
		\item $FT^k=\bra{\hat{k}}_{Q_2}=\sum_{j=1}^{d}  \frac{e^{2 \pi \im k j /d} }{\sqrt{d}} \bra{j}_{Q_2}$, ($1\leq k\leq d$)
		\item If $k=d$: stop. 
		
			Else: redefine $U=U\circ \sum_j e^{-2 \pi \im k j / d} \ketbra{j}{j}$ and restart with item 1.
	\end{enumerate}
 The probability of success in each round is given by $1/d$, independently of the initial state. After $M$ iterations we therefore have a probability of success given by $\sum_{i=1}^M 1/d (1-1/d)^{(i-1)}=1-\left(1-\frac{1}{d}\right)^M\rightarrow 1$ for $M\rightarrow\infty$.
\end{proof}
The next two pages are devoted to the proof of the closed form of $LOP$ operations, which simplifies its use for both theoretical as well as practical purposes, for instance for the comparison of $LOP$ to other resource theories.
\begin{prop*}[\ref{prop:form}]
	Let $\Lambda$ be a CPTP map acting on $W\otimes Q$. $\Lambda$ is an LOP operation exactly if it can be written as a finite sequence of maps $\Lambda=\sum_{\vec{\alpha}_{N}}\Lambda_{N}^{\alpha_N}(\vec{\alpha}_{N-1}) \circ\ldots \circ\Lambda_1^{\alpha_1}$, (where $N \leq \dim(W)$ is the length of the protocol, possibly depending on the previous outcomes $\vec{\alpha}_{t}=(\alpha_1,\ldots,\alpha_{t})$, 	$\vec{\alpha}_{0}=\vec{\alpha}_{-1}=0$) with CPTP maps $\Lambda_t(\vec{\alpha}_{t-1})=\sum_{\alpha_{t}}\Lambda_t^{\alpha_t}(\vec{\alpha}_{t-1})$ 
having a Kraus decomposition $ \Lambda_t(\vec{\alpha}_{t-1})[\rho]=\sum_{\alpha_{t}} K_t^{\alpha_{t}}(\vec{\alpha}_{t-1})\mdag \rho K_t^{\alpha_{t}}(\vec{\alpha}_{t-1})$ of the form
\begin{align}\label{eq:forma}
K_t^{\alpha_{t}}(\vec{\alpha}_{t-1})
=
\sum_{i=1}^{r(\vec{\alpha}_{t-2})}  \ketbra{\sigma_{\vec{\alpha}_{t}} \circ \min[r(\vec{\alpha}_{t-1}),i]}{i}_W
\otimes E_t^{ \alpha_{t}}\left(\vec{\alpha}_{t-1},i\right),
\end{align}
with $\sigma_{\vec{\alpha}_{t}}$ being an injective map to the positive labels of a new incoherent basis (see footnote~\cite{Note1}) and $E_t^{ \alpha_{t}}\left(\vec{\alpha}_{t-1},i\right)$ being an arbitrary operator acting on $Q$, potentially depending on previous outcomes and controlled by the populations of the wire and $1\leq r(\vec{\alpha}_{t})<r(\vec{\alpha}_{t-1})$ for any $t\in \{1,\ldots, N-1\}$ and $r(0)=\dim(W)$.
\end{prop*}

One of the non-trivial results that we need to show in the proof of prop~\ref{prop:form} is that classical to quantum forward can indeed be decomposed as described in the proposition. In the following lemma we slightly generalize this statement, as it does not significantly complicate the proof and the statement might be of independent interest.
\begin{lem}\label{lem:inj_ops}
	Let $\Lambda=\sum_{\alpha} F^{\alpha}\;\cdot \;{F^{\alpha}}\mdag$ be a CPTP map acting on both $W$ and $Q$ with 
	\begin{equation}\label{eq:simpleform}
	F^{\alpha}=\sum_{i=1}^d \ketbra{f(i)}{i}_W\otimes E^{\alpha}_Q(i),
	\end{equation}
	where $f%:\{1,...,d\}\rightarrow\{1,...,d\}
	$ maps indices to indices.
	Then $\Lambda$ admits a decomposition as in Prop.~\ref{prop:form}.
\end{lem}
\begin{proof}
	 For simplicity, we will only prove that there is a finite protocol of the given form. That the length of the protocol is bounded by the dimension of $W$ will be proven independently later in the proof of Prop~\ref{prop:form}.
	The function $f$ in Eq.~\ref{eq:simpleform} can map different members of the incoherent basis to the same output state. The main idea of the proof is to first reorder the incoherent basis of the wire (using a bijection) in such a way, that we can then use a sequence of maps with Kraus operators of the form Eq.~\ref{eq:form} to implement the same map and the same subselection possibilities as with the given operators~\ref{eq:simpleform}. The main trick is to iteratively collapse the subspaces belonging to the pre-image of $\ket{f(i)}$. 
	
	Let us begin with the case that the image of $f(i)$ is $\{1,\ldots s\}$ for a $s\in \mathbb{N} \leq \dim(W)$. Define $W_k=\{i\mid f(i)=k\}$ and a permutation $\sigma_1$ that maps the elements of $W_k$ to $\{1+\sum_{j=1}^{k-1}|W_j|,\ldots,\sum_{j=1}^{k}|W_j|\}$. This implements the announced reordering of the incoherent basis of the wire and corresponds to a unitary $\Lambda_1$ given by
 $K_1=\ketbra{\sigma_1(i)}{i}_W$. Next we define $r_t=t+\sum_{j=1}^{s-t} |W_j|$, e.g. $r_0=\dim(W)$, $r_1=1+\sum_{j=1}^{s-1}|W_j|$, $r_{s-1}=s-1+|W_1|$ and $r_{s}=s$. With this, we then define (for $t\in \{2,\ldots s+1\}$)

	\begin{equation}
		K_t^{\alpha_{t}}=\sum_{i=1}^{r_{t-2}} \ketbra{\sigma_{\oplus}^{r_{t-1}}\circ \min[{r_{t-1}},i]}{i}_W
		\otimes
		\left\{
		\begin{array}{ll}
		E^{\alpha_{t}}_{Q_1}(i)\otimes\ket{\alpha_{t}}_{Q_2}, &i\geq r_{t-1}\\
		\id_{Q_1}, & i<r_{t-1}
		\end{array}
		\right.
	\end{equation}
	where the permutation $\sigma_{\oplus}^{l}$ is defined by the mapping $i\mapsto i+1$, for $i< l$ and $l\mapsto 1$. These Kraus operators are of the form given in Eq.~\ref{eq:form}.

	From the CPTP condition for the Kraus operators defined in Eq.~\ref{eq:simpleform}, we get that $\sum_{i,j \in W_k} \ketbra{i}{j}_W\otimes \sum_{\alpha} {E^{\alpha}_Q(i)}\mdag E^{\alpha}_Q(j)=\sum_{i\in W_k} \ketbra{i}{i}_W\otimes \sum_{\alpha} {E^{\alpha}_Q(i)}\mdag E^{\alpha}_Q(i)=\id$. This implies that for each $t\in \{2,\ldots s+1\}$, $K_t^{\alpha_{t}}$ are Kraus operator of a CPTP map, that is $\sum_{\alpha_{t}} {K_t^{\alpha_{t}}}\mdag K_t^{\alpha_{t}}=\id$.
	
	It is straightforward to see by induction that 
	\begin{align*}
	K_t^{\alpha_{t}}\circ \ldots \circ K_1=&\sum_{i\in W_{s-t+2}} \ketbra{1}{i}_W\otimes E^{\alpha_{t}}_{Q_1}(i)\otimes\ket{\alpha_{t}}_{Q_2} \\
	&+\sum_{i\in W_{s-t+3}} \ketbra{2}{i}_W\otimes E^{\alpha_{t-1}}_{Q_1}(i)\otimes\ket{\alpha_{t-1}}_{Q_2} \\
	&+\ldots\\
	&+\sum_{i\in W_s} \ketbra{t-1}{i}_W\otimes E^{\alpha_{2}}_{Q_1}(i)\otimes\ket{\alpha_{2}}_{Q_2} \\
	& +\sum_{i\in \bigcup_{j=1}^{s-t+1} W_s} \ketbra{\sigma_1 (i)}{i}_W\otimes\id_Q \\
	=&\sum_{u=2}^t \sum_{i\in W_{s-t+u}} \ketbra{u-1}{i}_W \otimes E_{Q_1}^{\alpha_{t+2-u}} (i) \otimes \ket{\alpha_{t+2-u}} \\
	& +\sum_{i\in \bigcup_{j=1}^{s-t+1} W_s} \ketbra{\sigma_1 (i)}{i}_W\otimes\id_Q 
	\end{align*}
	Now we redefine the map $\Lambda_{s+1}\rightarrow \tr_{Q_2}\Lambda_{s+1}$, from which follows the statement in the case that the image of $f(i)$ is $\{1,\ldots s\}$ for a $s\in \mathbb{N}\le \dim(W)$. Now assume that the image of $f$ is not $\{1,\ldots s\}$. In this case, we can proceed in the same way and change the permutation of $\Lambda_{s+1}$ (before we trace out $Q_2$) such that we implement the correct $f$. 
\end{proof}
We are now ready to give the proof of Prop.~\ref{prop:form}.
\begin{proof}[Proof of Prop.~\ref{prop:form}]
	We need to prove four statements:
	\begin{enumerate}
		\item Any elemental LOP map can be decomposed into an arbitrarily long sequence of CPTP maps represented by Kraus operators of the form given in Eq.~\ref{eq:form}.
		\item Any CPTP map given by Kraus operators as in Eq.~\ref{eq:form} can be decomposed as an LOP map.
		\item Induction +1: The composition of two CPTP maps that can be decomposed into an arbitrary long sequence of CPTP maps  represented by Kraus operators of the form given in Eq.~\ref{eq:form} can again be decomposed into this form. This statement is trivial.
		\item Any CPTP map that can be decomposed into an arbitrary long sequence of CPTP maps represented by Kraus operators of the form given in Eq.~\ref{eq:form} can also be decomposed into such a sequence with $r(\vec{\alpha}_{t+1})< r(\vec{\alpha}_{t})$. From this follows that the choice $N\le \dim(W)$ is always possible, since $r(0)$ is w.l.o.g  equal to $\dim W$.
	\end{enumerate}
	The first statement is easy to see; permutations of the basis ${\cal Z}$ of $W$, diagonal unitaries on $W$, and observed quantum operations on $Q$ all have the form of $K_1^{\alpha_1}$. That classical to quantum forwarding has the form~\ref{eq:form} is a direct corollary of Lem.~\ref{lem:inj_ops}.
	
	To implement a map given by Kraus operators of the form~\ref{eq:form} ($K_t^{\alpha_{t}}(\vec{\alpha}_{t-1})
	=
	\sum_{i=1}^{r(\vec{\alpha}_{t-2})}  \ketbra{\sigma_{\vec{\alpha}_{t}} \circ \min[r(\vec{\alpha}_{t-1}),i]}{i}_W
	\otimes E_t^{ \alpha_{t}}\left(\vec{\alpha}_{t-1},i\right)$) by elemental operations, the first step is to implement the trivial observed quantum operation 
	\begin{align*}
	\ket{0}_{W_2},
	\end{align*} 
	followed by a permutation on $W$ given by 
	\begin{align*}
		\sum_{i=1}^{r(\vec{\alpha}_{t-2})}  \ketbra{ \min[r(\vec{\alpha}_{t-1}),i]}{i}_{W_1}  \otimes\ketbra{i}{0}_{W_2}.
	\end{align*} 
	Then we do a classical to quantum forwarding of system $W_2$ to an ancillary system $Q_2$. Up to here we can summarize the concatenation of these operations by 
	\begin{align*}
		\sum_{i=1}^{r(\vec{\alpha}_{t-2})}  \ketbra{\min[r(\vec{\alpha}_{t-1}),i]}{i}_{W_1}  \otimes \id_{Q_1} \otimes\ket{i}_{Q_2}.
	\end{align*}
	 The next step is a quantum operation on $Q$ defined by the Kraus operators
	\begin{align*}
	\sum_{i=1}^{r(\vec{\alpha}_{t-2})} E_t^{ \alpha_{t}}\left(\vec{\alpha}_{t-1},i\right)_{Q_1} 
	\otimes   \ketbra{\sigma_{\vec{\alpha}_{t}} \circ \min[r(\vec{\alpha}_{t-1}),i]}{i}_{Q_2},
	\end{align*}
	which is a quantum operation exactly if the Kraus operators $K_t^{\alpha_{t}}(\vec{\alpha}_{t-1})
	$ form one. 
	In total, we then have
	\begin{align*}
	\sum_{i=1}^{r(\vec{\alpha}_{t-2})}  \ketbra{\min[r(\vec{\alpha}_{t-1}),i]}{i}_{W_1} 
	 \otimes  E_t^{ \alpha_{t}}\left(\vec{\alpha}_{t-1},i\right)_{Q_1}  
	\otimes   \ket{\sigma_{\vec{\alpha}_{t}} \circ \min[r(\vec{\alpha}_{t-1}),i]}_{Q_2}.
	\end{align*}
	After a permutation on $W_1$ that implements $\sigma_{\vec{\alpha}_{t}}$, we obtain in total
	\begin{align*}
	\sum_{i=1}^{r(\vec{\alpha}_{t-2})}  \ketbra{\sigma_{\vec{\alpha}_{t}} \circ \min[r(\vec{\alpha}_{t-1}),i]}{i}_{W_1} 
	\otimes  E_t^{ \alpha_{t}}\left(\vec{\alpha}_{t-1},i\right)_{Q_1}  \\
	\otimes   \ket{\sigma_{\vec{\alpha}_{t}} \circ \min[r(\vec{\alpha}_{t-1}),i]}_{Q_2}.
	\end{align*}
	The last step is to use the operation ${\cal B}_{W_1, Q_2}^{-1}$ from Lem.~\ref{lem:bij} to get rid of $Q_2$ and we end up with the wanted operation.

	As already mentioned, the third statement is trivial.
	The hard part is the fourth statement. Assume we have two arbitrary sets of Kraus operators $K_{t}^{\alpha_{t}}$, $K_{t-1}^{\alpha_{t-1}}$ of the given form corresponding to $\Lambda_t$ and $\Lambda_{t-1}$. Then we distinguish two cases. In the first case, we show that we can find two new sets of Kraus operators $L_{t}^{\alpha_{t}}$, $L_{t-1}^{\alpha_{t-1}}$ of the required form such that $K_{t}^{\alpha_{t}} \circ K_{t-1}^{\alpha_{t-1}}= L_{t}^{\alpha_{t}} \circ L_{t-1}^{\alpha_{t-1}}$, $r(\vec{\alpha}_{t-2})$ remains the same and $r(\vec{\alpha}_{t-1})< r(\vec{\alpha}_{t-2})$ for the two new sets (the place where one cuts the Hilbert space dimension of the wire in the step $t$ ($r(\vec{\alpha}_{t-1})$) depends on the previous outcomes, but not on the current one. That is why its index is $t-1$ and not $t$). In the second case, one can replace the two CPTP maps by one CPTP map of the required form such that $r(\alpha_{t-2})$ remains unchanged.

	Assume $r(\vec{\alpha}_{t-1})\geq r(\vec{\alpha}_{t-2})$ (otherwise there is nothing to show). First we split up the injection $\sigma_{\vec{\alpha}_{t-1}}(i)$ into a permutation and an order-preserving injection. Formally, we define the permutation $\eta_{\vec{\alpha}_{t-1}}$ on $\{1,\ldots r(\vec{\alpha}_{t-2})\}$ and the injection $a:\{1,\ldots r(\vec{\alpha}_{t-2})\}\rightarrow\sigma_{\vec{\alpha}_{t-1}}(\{1,\ldots r(\vec{\alpha}_{t-2})\})\subset \mathbb{N}^{>0}$  such that 
	\[
	\sigma_{\vec{\alpha}_{t-1}}(i)=a(\eta_{\vec{\alpha}_{t-1}}(i)),\] and $a(i)<a(j)$ for $i<j$. Then there is some $l\leq r(\vec{\alpha}_{t-2})$ such that 
	\[\min[r(\vec{\alpha}_{t-1}),\cdot]\circ \sigma_{\vec{\alpha}_{t-1}}(\{1,\ldots r(\vec{\alpha}_{t-2})\})=\{a(1),\ldots a(l)\}\] 
	and 
	\[\min[r(\vec{\alpha}_{t-1}),\cdot] \circ\sigma_{\vec{\alpha}_{t-1}}
	=\min[r(\vec{\alpha}_{t-1}),\cdot] \circ a \circ\eta_{\vec{\alpha}_{t-1}}=a\circ \min[l,\cdot] \circ\eta_{\vec{\alpha}_{t-1}},\]
	 with $f(x,\cdot)$ denoting the function $f(x,y)$ for fixed $x$, as a function of $y$ and we use $f(A)$ for a function $f$ and a set $A$ to denote the image of the set $A$ under $f$.
	
	We first consider the case $l<r(\vec{\alpha}_{t-2})$. We first define \[F_{t}^{\alpha_{t}}(\vec{\alpha}_{t-1}, i)=E_{t}^{\alpha_{t}}(\vec{\alpha}_{t-1}, \sigma_{\vec{\alpha}_{t-1}}\circ \eta_{\vec{\alpha}_{t-1}}^{-1}(i))
	=E_t^{\alpha_{t+1}}(\vec{\alpha}_t, a(i)),\]
	with $i\in \{1,\ldots r(\vec{\alpha}_{t-2})\}$. By further defining the injection $\eta_{\vec{\alpha}_{t}}(k)=\sigma_{\vec{\alpha}_{t}} (a(k))$, we can finally define:
	\[
	L_t^{\alpha_{t}}(\vec{\alpha}_{t-1})
	=\sum_{i=1}^{r(\vec{\alpha}_{t-2})}
	\ketbra{\eta_{\vec{\alpha}_{t}}\circ \min[l,i]}{i}_W\otimes F_{t}^{\alpha_{t}}(\vec{\alpha}_{t-1}, i)
	\]
	and
	\[
	L_{t-1}^{\alpha_{t-1}}(\vec{\alpha}_{t-2})
	=\sum_{i=1}^{r(\vec{\alpha}_{t-3})}
	\ketbra{\eta_{\vec{\alpha}_{t-1}}\circ \min[r(\vec{\alpha}_{t-2}),i]}{i}_W\otimes E_{t-1}^{\alpha_{t-1}}(\vec{\alpha}_{t-2}, i).
	\]
	First we note that the map defined by $L_{t-1}^{\alpha_{t-1}}(\vec{\alpha}_{t-2})$ is CPTP exactly if the map defined by $K_{t-1}^{\alpha_{t-1}}(\vec{\alpha}_{t-2})$ is (as they only differ by a permutation at the end).
	
	Now we show that $L_t^{\alpha_{t}}(\vec{\alpha}_{t-1})$ forms a CPTP map as well. Remember that $\min[r(\vec{\alpha}_{t-1}),a(\cdot)] =a( \min[l,\cdot])$ and since $a$ is a bijection on its image,  $\braket{\min[r(\vec{\alpha}_{t}),a(i)]}{\min[r(\vec{\alpha}_{t}),a(j)]}=\braket{\min[l,i]}{\min[l,j]}$. %Using this fact and, we show that indeed $L_t^{\alpha_{t+1}}(\vec{\alpha}_{t})$ forms a CPTP map provided $K_t^{\alpha_{t+1}}(\vec{\alpha}_{t})$ forms one.
	Then
		\begin{align*}
			\id&=
			\sum_{\alpha_{t}} K_t^{\alpha_{t}}(\vec{\alpha}_{t-1})	\mdag 
			  K_t^{\alpha_{t}}(\vec{\alpha}_{t-1})\\
			\Leftrightarrow  &
			\forall i,j: 
			\\
			\id\delta_{i,j}&=(\bra{i}_W\otimes\id_Q)  
			\left(	\sum_{\alpha_{t}} K_t^{\alpha_{t}}(\vec{\alpha}_{t-1})	\mdag 
			K_t^{\alpha_{t}}(\vec{\alpha}_{t-1})\right)   (\ket{j}_W\otimes\id_Q)\\
			&=\sum_{\alpha_{t}} \braket{\min[r(\vec{\alpha}_{t-1}),i]}{\min[r(\vec{\alpha}_{t-1}),j]} 
			E_t^{\alpha_{t}}(\vec{\alpha}_{t-1}, i)\mdag E_t^{\alpha_{t}}(\vec{\alpha}_{t-1}, j)\\
			\Leftrightarrow &\forall i,j: 
			\\
			\;\id\delta_{i,j}&=\sum_{\alpha_{t}} \braket{\min[r(\vec{\alpha}_{t}),a(i)]}{\min[r(\vec{\alpha}_{t}),a(j)]} \\
			& \;\;\;\;\;\;\;\;\;\; E_t^{\alpha_{t}}(\vec{\alpha}_{t-1}, a(i))\mdag E_t^{\alpha_{t}}(\vec{\alpha}_{t-1}, a(j))\\
			&=\sum_{\alpha_{t}} \braket{\min[l,i]}{\min[l,j]} 
			F_{t}^{\alpha_{t}}(\vec{\alpha}_{t-1}, i)\mdag F_{t}^{\alpha_{t}}(\vec{\alpha}_{t-1}, j)\\	
			\Leftrightarrow  &\id=
			\sum_{\alpha_{t}} L_t^{\alpha_{t}}(\vec{\alpha}_{t-1})	\mdag 
			  L_t^{\alpha_{t}}(\vec{\alpha}_{t-1}).
		\end{align*}
		where we used in the last line that $\eta_{\vec{\alpha}_{t}}$ is a bijection as well.

	Finally we need to check that we get indeed the right mapping ($K_t^{\alpha_{t}}(\vec{\alpha}_{t-1})\circ 	K_{t-1}^{\alpha_{t-1}}(\vec{\alpha}_{t-2})=L_t^{\alpha_{t}}(\vec{\alpha}_{t-1})\circ 	L_{t-1}^{\alpha_{t-1}}(\vec{\alpha}_{t-2})$). This follows from the equalities %(for $j\in \{1,\ldots , r(\vec{\alpha}_{t-2})\}$)
	\begin{align*}
	&\sum_{i=1}^{ r(\vec{\alpha}_{t-2})} \ketbra{\sigma_{\vec{\alpha}_{t}}\circ \min[r(\vec{\alpha}_{t-1}),i]}{i}\circ \ket{\sigma_{\vec{\alpha}_{t-1}}(j)}
	\otimes E_{t}^{\alpha_{t}}(\vec{\alpha}_{t-1}, i)\circ E_{t-1}^{\alpha_{t-1}}(\vec{\alpha}_{t-2}, j)
	\\
	&= \ket{\sigma_{\vec{\alpha}_{t}}\circ \min[r(\vec{\alpha}_{t-1}),\sigma_{\vec{\alpha}_{t-1}}(j)]}
	\otimes E_{t}^{\alpha_{t}}(\vec{\alpha}_{t-1},\sigma_{\vec{\alpha}_{t-1}}(j))\circ E_{t-1}^{\alpha_{t-1}}(\vec{\alpha}_{t-2}, j)
	\\
	&= \ket{\sigma_{\vec{\alpha}_{t}}\circ a\circ \min[l,  \eta_{\vec{\alpha}_{t}}(j)]}
	\otimes E_{t}^{\alpha_{t}}(\vec{\alpha}_{t-1}, a(\eta_{\vec{\alpha}_{t-1}}(j)))\circ E_{t-1}^{\alpha_{t-1}}(\vec{\alpha}_{t-2}, j)\\
	&= \ket{\eta_{\vec{\alpha}_{t}}\circ \min[l,\eta_{\vec{\alpha}_{t-1}}(j)]}
	\otimes F_{t}^{\alpha_{t}}(\vec{\alpha}_{t-1},\eta_{\vec{\alpha}_{t-1}}(j)))\circ E_{t-1}^{\alpha_{t-1}}(\vec{\alpha}_{t-2}, j)\\
	&=\sum_{i=1}^{ r(\vec{\alpha}_{t-2})} \ketbra{\eta_{\vec{\alpha}_{t}}\circ \min[l,i]}{i}\circ \ket{\eta_{\vec{\alpha}_{t-1}}(j)}
	\otimes F_{t}^{\alpha_{t}}(\vec{\alpha}_{t-1},i))\circ E_{t-1}^{\alpha_{t-1}}(\vec{\alpha}_{t-2}, j).
	\end{align*}
	The case $l=r(\vec{\alpha}_{t-1})$ can be handled similarly, just by noting that the action of $\cut_l$ gets trivial and one can therefore express the concatenation of the two maps as a single map of the same form.
	
	This proves that one can chose $r$ such, that: $\dim {W} \geq r(0) > r(\vec{\alpha}_{1}) > \ldots > r(\vec{\alpha}_{N-1})\geq 1$. It follows that it is always possible to have $N\leq \dim{W}$.
\end{proof}
We can now use the above proposition to prove one direction of the connection between $LOP$ and $SIO$ operations given in proposition~\ref{prop:PIO}.
\begin{prop*}[\ref{prop:PIO}]
	Let $\Lambda$ be a CPTP map acting on $W\otimes Q$. $\Lambda$ is an LOP operation exactly if it can be written as a sequence of maps $\Lambda=\Lambda_M(\vec{\alpha}_{M-1}) \circ\ldots \circ\Lambda_1$, for some finite $M$,
	where each $\Lambda_i$ is
	\begin{itemize}
		\item[(a)]  a physical incoherent operation on $W$
		\item[or (b)] a destructive measurement in one fully coherent basis of a subsystem of $W$
		\item[or (c)]	a controlled unitary (${\bf U}_{\rm control}=\sum_m |m\rangle\langle m|_{\rm control} \otimes {\bf U}_{{\rm target}}(m)$) with control $W$ and target $Q$ 
		\item[or (d)] a generalized measurement of $Q$, encoding the result on $W$ ($\rho \mapsto \sum \limits_{\alpha} \ketbra{\alpha}{\alpha}_W\otimes K^{\alpha}\rho {K^{\alpha}}\mdag$).
	\end{itemize}
	One can equivalently replace item (a) by ``a special incoherent operation on $I$''.
\end{prop*}
\begin{proof}
	We start the proof by noting that any destructive measurements on a subsystem of $W$ can be done by LOP operations by classical to quantum forwarding of the subsystem in question to an ancilla in $Q$, $Q_2$ and then doing the measurement there. Special incoherent operations are those with a Kraus decomposition with Kraus operators of the form $K^{\alpha_1} =
	\sum_{i=1}^{d}  c(\alpha)\ketbra{\sigma_{\alpha }( i)}{i}$~\cite{Winter2016}, which obviously is a special case of the form $K_1$ in Prop.~\ref{prop:form} (i.e. Eq.~\ref{eq:form}), the same is true for control unitaries of the form (c). Next we note that PIO is a subset of SIO~\cite{Chitambar2016b}, so that any operation having a decomposition as in the proposition is an element of LOP.
	
	For the converse we only need to show that we can do classical to quantum forwarding using only the operations (a)-(d). This goes by virtually the same protocol we already applied for the inverse part of Lem.~\ref{lem:bij}.
	We first do a controlled unitary from the system $W_s$ in question to an ancillary system $Q_s$ prepared in the state $\ket{0}$, to which we want to teleport, apply a measurement in the Fourier basis of $W_s$ (note that measurements in different fully coherent bases only differ by a diagonal unitary, which is an element of PIO, so we can assume w.l.o.g that the basis we can measure in is given by the Fourier basis), followed by a correction of the phase on $Q_s$:
	\begin{enumerate}
		\item $\ket{0}_{Q_s}$
		\item $\sum_{i}\ketbra{i}{i}_{W_s} \otimes  \sum_{j}\ketbra{i\oplus j}{j}_{W_s}$
		\item	$FT^k_{W_s}=\bra{\hat{k}}_{W_s}=\sum_{j=1}^{\dim(W_s)}  \frac{e^{2 \pi \im k j /\dim(W_s)} }{\sqrt{\dim(W_s)}}\bra{j}_{W_s}$
		\item  $D(k)= \sum_{l=1}^{\dim(W_s)} \ketbra{l}{l}_{Q_s}   e^{-2 \pi \im k l / \dim(W_s)} $,
	\end{enumerate}
	resulting in the action $\sum_i \bra{i}_{W_s}\otimes  \ket{i}_{Q_s}$.
\end{proof}
The most general set of meaningful operations, if one just has the restriction that one wants to keep the wire incoherent and classically correlated with the quantum system, are the operations that do not allow to generate states which are not incoherent-quantum from those that are. If one allows post-selection this means that the same should be true for each Kraus operator defining the operation. This set, introduced in~\cite{Ma2016}, is called $IQO$, for incoherent-quantum operations. Our approach by contrast, is operational, in the sense that we only allow some specific elemental operations which are meaningful. This has the advantage that it is more transparent and does not allow operations that could be not classical in a way that might not be obvious, but the disadvantage, that there might be some operations that we do not allow, that are still meaningful. Fortunately we can prove that the gap between the theory we propose and the completely abstractly defined maximal possible theory $IQO$ is not that big. The details of this statement is the content of proposition~\ref{prop:gennoninj}.
\begin{prop*}[\ref{prop:gennoninj}]
	Be $\Lambda$ an incoherent-quantum operation on $W_1\otimes Q$, which is exactly the case if it is CPTP and has a Kraus decomposition with Kraus operators of the form $K^{\alpha}=\sum_i \ketbra{f_{\alpha}(i)}{i}_{W_1} \otimes E^{\alpha}(i)$, for some functions $f_{\alpha}$ acting on the labels of the incoherent basis and some operators $E^{\alpha}(i)$ acting on $Q$. Let $d := \dim(W_1)$. If $d=2$, $\Lambda \in LOP$. For $d\geq 3$ $LOP\neq IQO$, but there is a stochastic implementation of the map in $LOP$ with 
	a success rate of at least
	$1/d$; i.e. there is an operation $\Lambda' \in LOP(W_2\otimes W_1\IC Q)$  with $\Lambda'[\ketbra{0}{0}_{W_2}\otimes\rho]=\ketbra{0}{0}_{W_2}\otimes\Lambda_0'[\rho]+\ketbra{1}{1}_{W_2}\otimes\Lambda_1'[\rho]$ with $\Lambda_1'[\rho]=\Lambda[ \rho]/d\; \forall \rho$.
\end{prop*}
\begin{proof}
	The form of the incoherent-quantum operations directly follows from applying each Kraus operators to a product state incoherent on $W$ and requiring that it is still (up to normalization) a product state incoherent on $W$. The converse, namely that any Kraus operator of that form is incoherent-quantum (that is preserves the set of incoherent-quantum states) is a trivial consequence of the convexity of the set of incoherent-quantum states.
	
	The protocol for doing the operation given by the Kraus operators $K^{\alpha}=\sum_i \ketbra{f_{\alpha}(i)}{i}_{W} \otimes E_{\alpha}(i)_Q$, by LOP operations with a probability of $1/d$ is given by (identifying $W=W_1$ and $Q=Q_1$ at the beginning and the end):
	\begin{enumerate}
		\item $\sum_i \ketbra{i}{i}_{W_1}\otimes\ket{i}_{Q_2}$ (see Lem.~\ref{lem:bij}),
		\item $\sum_i E_{\alpha}(i)_{Q_1} \otimes \ketbra{f_{\alpha}(i)}{i}_{Q_2}$,
		\item $\sum_i \ketbra{f_{\alpha}(i)}{i}_{W_1}\otimes\ket{i}_{W_2}$ (a permutation after adding an ancilla),
		\item "Delete" the duplicate $Q_2$, (applying ${\cal B}^{-1}$ of Lem.~\ref{lem:bij} to $W_1$, $Q_2$),
		\item $\sum_i \bra{i}_{W_2}\otimes\ket{i}_{Q_2}$,
		\item $\bra{\hat{k}}_{Q_2}$. 
	\end{enumerate}
In total the operation is given by the Kraus operators $\sum_i \ketbra{f_{\alpha}(i)}{i}_W \otimes E_{\alpha}(i)_{Q} \cdot\braket{\hat{k}}{i}$.
If the outcome is $k=0$, $\braket{\hat{k}}{i}=1/\sqrt{d}$ $\forall i$ and the protocol is successful. Note that the probability for this is $1/d$ independently of the initial state the operation is applied to. If $k\neq 0$ there is an $i$-dependent phase and in general the protocol fails (the information about $i$ is lost, so that at this point there is no way to correct the phases).

Let's consider the case $d=2$. We define the set $R=\{\alpha\mid f_{\alpha}(1)=f_{\alpha}(2)\}$ and $R^c$ its complement, separating the injective from the non-injective functions on $W$. The idea in the following is to first check (on $Q$) whether one has an injective or a non-injective case on $W$ and then change the form on $W$ accordingly, while inverting the check and applying the final operation on $Q$. Let wlog. $\alpha\in \{1,\ldots N\}$. We note that since the $K^{\alpha}$  form a CPTP map we have in particular that $\sum_{\alpha\in R} E_{\alpha}(1)_{B}\mdag \circ E_{\alpha}(2)_{B}=0$. For this reason it makes sense to define the operations $E^0(i)=\sqrt{ \sum_{\alpha\in R} E_{\alpha}(i)_{B}\mdag \circ E_{\alpha}(i)_{B}}$ and $K^0=\sum_i \ketbra{i}{i}\otimes E^0(i)$ and it is easy to check that $K^{\alpha}$ with $\alpha\in \{0\}\cup R^c$, again forms a CPTP map. This map has the form of Eq.~\ref{eq:form} and is therefore an element of $LOP$. For $\alpha \in R$ we then need a second step and we define the operations $E_1^{\alpha}(i)=E^{\alpha}(i)\circ E^0(i)^{-1}$, where $E^0(i)^{-1}$ is the Penrose pseudo-inverse of $E_0(i)$, that is, if $E_0(i)=U(i)\circ  D(i) \circ U(i)\mdag$ is the singular value decomposition of $E_0(i)$ (where we used that $E_0(i)=E_0(i) \mdag$, from the definition of $E_0(i)$), then $E^0(i)^{-1}=U(i) \circ  D(i)^{-1} \circ U(i)\mdag$ (here $D(i)$ is a diagonal matrix and $D(i)^{-1}$ is its diagonal right-inverse on its support and vice versa). We also need the operation $E_1^{0}(i)=\ket{i}\otimes(\id-E^0(i)\circ E^0(i)^{-1})$. Here it is useful to note that  $E^0(i)\circ E^0(i)^{-1}=E^0(i)^{-1}\circ E^0(i)$ is the projection on the image of $E^0(i)$. We can then define in the notation of Prop.~\ref{prop:form} $K_1^{\alpha}= \sum_i \ket{1}\bra{i} \otimes E_1^{\alpha}(i)$ for $\alpha \in R\cup \{0\}$ ($f_1(i)=1$).
We then have that
\begin{align*}
	&\sum_{\alpha} {K_1^{\alpha}}\mdag K_1^{\alpha}\\
	&=\sum_{\alpha} \sum_{i,j} \ketbra{i}{j}\otimes E_1^{\alpha}(i)\mdag E_1^{\alpha}(i)\\
	&=\sum_{i,j} \ketbra{i}{j}\otimes E^0(i)^{-1} \stackrel{\delta_{i,j} E^0(i)^2}{\overbrace{\sum_{\alpha\in R} \left(E^{\alpha}(i)\mdag E^{\alpha}(j)\right)}} E^0(j)^{-1}
	\\&+\sum_{i,j} \ketbra{i}{j}\otimes  \braket{i}{j} (\id-E^0(i)\circ E^0(i)^{-1})\mdag (\id-E^0(i)\circ E^0(i)^{-1})
	\\
	&=\sum_{i} \ketbra{i}{i}\otimes E^0(i)^{-1}E^0(i)\\
	&+ \sum_{i} \ketbra{i}{i}\otimes(\id -2 E^0(i)E^0(i)^{-1} +E^0(i)E^0(i)^{-1})\\
	&=\id
\end{align*}
Noticing that the probability to measure $\alpha=0$ in the second step, provided that the outcome in the first was $\alpha=0$, is $0$ and that $E^{\alpha}(i)\circ E^0(i)^{-1}\circ E^0(i)=E^{\alpha}(i)$ (since the support of $E^0(i)$ contains the support of $E^{\alpha}(i)$), we find that indeed applying the protocol as in Prop.~\ref{prop:form} with the above defined operations yields the right map.

The proof of the statement $LOP\neq IQO$ is done in section~\ref{sec:counterexample}, where we provide an explicit counterexample for a wire with Hilbert-space dimension $3$.
\end{proof}
Whether one can meaningfully call states that are not free in a given resource theory "resources", depends on whether they can be used to do tasks that are impossible under the application of free operations alone.
Theorem~\ref{prop:resource} shows that maximally coherent states are resources in the maximal sense of the word: they enable to do anything within quantum mechanics.
\begin{thm*}[\ref{prop:resource}]
		Let $\Lambda$ be a CPTP map acting on $W_1\otimes Q$, with $W_1$ having dimension $d$. Let $\ket{\psi} = \sum_{i=1}^d \frac{1}{\sqrt{d}} \ket{i}$ be a maximally coherent state on $W_2$. Then there is an operation $\Lambda'\in LOP(W_2\otimes W_1\IC Q)$, with $\tr_{W_2} [\Lambda'[\ketbra{\psi}{\psi}_{W_2}\otimes\rho_{W_1,Q}]]=\Lambda[\rho_{W_1,Q}]$.
\end{thm*}
\begin{proof}
	The trick is to do the operation on the quantum side, that is: send the system $W_1$ to $Q_2$, do the operation $\Lambda$ on $Q_2,Q_1$. Then use Lem.~\ref{lem:bij} to construct a Bell-type state from the ancillary coherent state on $W_2$. Finally teleport the system $Q_2$ back to $W_1$ using the original teleportation protocol~\cite{RDF.03} and using up the ancillary Bell state. In detail (identifying $Q_1=Q$ at the beginning and at the end):
	\begin{enumerate}
		\item Preparation: $\ket{\psi}_{W_2}$,
		\item (free) teleportation to the quantum side $\sum_i \bra{i}_{W_1}\otimes  \ket{i}_{Q_2} $,
		\item Application of $\Lambda$ on the quantum side:  $\Lambda_{Q_2,Q_1}$,
		\item Doubling of $W_2$ (Lem.~\ref{lem:bij}): $\sum_i \ketbra{i}{i}_{W_2} \otimes  \ket{i}_{Q_3}$,
		\item Measurement in the `Bell basis' $\bra{b(k,l)}_{Q_3,Q_2}$, given by $\ket{b(k,l)}=CNOT\circ(\ket{\hat k}\otimes\ket{l})=1/\sqrt{d}\sum_{j} e^{2 \pi \im k j/d} \ket{j} \otimes \ket{l\oplus (j-1)}$, with $a\oplus b=\mod_d (a+b-1)+1$,
		\item Finish with a diagonal unitary $ \sum_j e^{2 \pi \im k j/d} \ketbra{j}{j}_{W_2}$ on $W_2$ followed by a permutation $\sum_j \ket{l\oplus(j-1)}_{W_1}\bra{j}_{W_2}$, both depending on the $d^2$ possible outcomes of the previous measurement, given by the indices $k,l$.
	\end{enumerate}
	Using the Kraus decomposition for the map $\Lambda(\rho)=\sum_{\alpha} K^{\alpha} \rho {K^{\alpha}}\mdag$, we get that the full protocol is given by the ($d^2$) Kraus operators 
	\begin{align*}
	&1/d \sum_{i,j} \ket{l\oplus (j-1)}\bra{i}_{W_1}
	\otimes \left((\bra{l\oplus (j-1)}_{Q_2}\otimes\id_{Q_1}) \circ K^{\alpha}\circ (\ket{i}_{Q_2}\otimes\id_{Q_1})
		\right)
	\\	&=K^{\alpha}/d.
\end{align*}
\end{proof}

\section{Proofs on the coherence cost of entanglement}
Both theorems on the coherence cost of entanglement presented in the main part depend heavily on the remark~\ref{rem:generalization}, which results from just following the respective protocols. To facilitate reading we repeat it here:
\begin{rem*} [\ref{rem:generalization}]
$\forall \rho_{W}=\sum_i p_i \ketbra{i}{i}_{W}$,
\begin{align*}
1. \;&\tr_{W} (\Lambda \circ (\rho_{W}\otimes\id_Q))\in LOCC({Q}^1,\ldots,{Q}^n) \\
&\forall \Lambda \in LOP({Q}^{i(1)} \CI {W}^{j(1)} \IC {Q}^{i(2)} \CI {W}^{j(2) }\IC \ldots),
\\
2. \;&\Lambda\otimes\id_{W} \in LOP({Q}^{i(1)} \CI {W}^{j(1)} \IC {Q}^{i(2)} \CI {W}^{j(2) }\IC \ldots)\\
&\forall \Lambda \in LOCC({Q}^1,\ldots,{Q}^n). 
\end{align*}
\end{rem*}
We then have that,
\begin{thm*}[\ref{thm:bipartite}]
	Let $\eta_{LOCC}=\sum_{ij} r_{i,j} \ketbra{i}{j}_{Q^1_1}\otimes\ketbra{i}{j}_{Q^2_1}$ be a maximally correlated state (in arbitrary orthonormal local bases of $Q^1_1\otimes Q^2_1$) and $\eta_{LOP}=\sum_{ij} r_{i,j} \ketbra{i}{j}_{{W}}$ be a corresponding state in the incoherent basis $\cal Z$ of ${W}$.
If $\Lambda$ is  a CPTP map on ${Q}^1\otimes {Q}^2$, then the following statements are equivalent:
	\begin{align*}
	1.\;\;& \exists 	\Lambda_{LOCC}\in LOCC({Q}^1,{Q}^2): \\
	&\Lambda_{LOCC} \left[ \eta_{LOCC} \otimes \rho_{Q^1_2,Q^2_2} \right] 	=\Lambda[\rho_{Q^1_2,Q^2_2}] \; \forall \rho_{Q^1_2,Q^2_2}\\
	2.\;\;& \exists \Lambda_{LOP}\in LOP({Q}^1\CI {W} \IC {Q}^2): 
	\\&  \Lambda_{LOP} \left[ \eta_{LOP} \otimes \rho_{Q^1_2,Q^2_2} \right] 	=\Lambda[\rho_{Q^1_2,Q^2_2}] \; \forall \rho_{Q^1_2,Q^2_2}.
	\end{align*}
\end{thm*}
\begin{proof}
	The $``\Rightarrow"$ statement is a direct corollary of Rem.~\ref{rem:generalization} and Lem.~\ref{lem:bij}. The operation $\cal B$ in Lem.~\ref{lem:bij} allows to transform the state $\sum_{ij} r_{i,j} \ketbra{i}{j}_{W} \otimes \rho_{Q^1_2,Q^2_2}$ to $\sum_{ij} r_{i,j} \ketbra{ii}{jj}_{W,Q^2_2} \otimes \rho_{Q^1_2,Q^2_2}$, which by classical to quantum forwarding can in turn be transformed to $\sum_{ij} r_{i,j} \ketbra{ii}{jj}_{Q^1_2,Q^2_2} \otimes \rho_{Q^1_1,Q^2_1}$ (by $LOP(Q^1\CI W \IC Q^2)$ operations). By the Rem.~\ref{rem:generalization} we then get that the $LOCC(Q^1,Q^2)$ operation that reproduces $\Lambda$ is also a $LOP(Q^1\CI W \IC Q^2)$ operation, which concludes the proof of this direction.
	
	For the converse, assume that a protocol is given, implementing $\Lambda$ by elemental operations of $LOP(Q^1\CI W \IC Q^2)$, using an ancillary state $\sum_{ij} r_{i,j} \ketbra{i}{j}_{W}$. To get the equivalent $LOCC(Q^1,Q^2)$ protocol using the ancillary state $\sum_{ij} r_{i,j} \ketbra{ii}{jj}_{{Q_2^1},{Q_2^2}}$, replace the elemental operations of the given protocol in the following way by $LOCC(Q^1,Q^2)$ operations:
	\begin{enumerate}
		\item \emph{Permutations:} $\sum_i \ketbra{\sigma(i)}{i}_{W}$, by permutations on both sides: $\sum_i \ketbra{\sigma(i)}{i}_{Q^1_2}$ followed by $\sum_i \ketbra{\sigma(i)}{i}_{Q^2_2}$.
		\item \emph{Phases:} $\sum_j e^{\im \phi(j)} \ketbra{j}{j}_{W}$, by phases on one side: $\sum_j e^{\im \phi(j)} \ketbra{j}{j}_{Q^1_2}$
		\item \emph{Observed quantum operations:} $\ket{\alpha}_{W_a}\otimes K^{\alpha}_{Q^s}$, by the same operation, with the outcome in an ancilla of the system $Q^s$, classically communicating the result $\alpha$ to the other side $Q^{\neg s}$and encoding it there as well: $\ket{\alpha}_{Q^s_a}\otimes K^{\alpha}_{Q^s}$ followed by $\ket{\alpha}_{Q^{\neg s}_a}$.
		\item \emph{Classical to quantum forwarding:} $\sum_{k} \bra{k}_{W_a}\otimes\ket{k}_{Q^s_a}$, by first deleting the copy on $Q_2^{\neg s}$ (by doing a Fourier measurement, followed by a correction of the phase on $Q_2^{s}$, as in the proof of Lem.~\ref{lem:bij}), followed by the trivial forwarding: $\sum_{k} \bra{k}_{Q^s_2}\otimes\ket{k}_{Q^s_a}$.
	\end{enumerate} 
	Note that by doing these replacements (and merging the ancillary Hilbert spaces with $Q^j_2$, $j=1,2$ after each step) a generic state $\sum_{i,j} s_{i,j} \ketbra{i}{j}_{W} \otimes {\tau(i,j)}_{Q^1_1,Q^2_1}$ in any step of the protocol gets mapped to $\sum_{i,j} s_{i,j} \ketbra{ii}{jj}_{{Q_2^1},{Q_2^2}} \otimes {\tau(i,j)}_{Q^1_1,Q^2_1}$, from which the assertion follows.
\end{proof}
A direct corollary of the remark~\ref{rem:generalization} is Thm.~\ref{thm:conv}, which introduces useful conditions for state transformations in multipartite entanglement.
\begin{thm*}[\ref{thm:conv}]
		Let $\rho$, $\sigma$ be states on ${W}$. \\If $\exists \tau_{W}$, a state on ${W}$, s.t. 
		$$\tau_{W} \stackrel{LOP({Q}^{i(1)} \CI {W}^{j(1)} \ldots)}{\rightarrow} \rho,$$ but $$\tau_{W} \stackrel{LOP({Q}^{i(1)} \CI {W}^{j(1)} \ldots)}{\nrightarrow} \sigma,$$ then it follows that
	$\rho \stackrel{LOCC({Q}^1,\ldots,{Q}^n)}{\nrightarrow} \sigma$
\end{thm*}
\begin{proof}
	Assume the theorem is not valid, that is, $$\tau_{W} \stackrel{LOP({Q}^{i(1)} \CI {W}^{j(1)} \ldots)}{\rightarrow} \rho,$$  $$\tau_{W} \stackrel{LOP({Q}^{i(1)} \CI {W}^{j(1)} \ldots)}{\nrightarrow} \sigma$$ and $$\rho \stackrel{LOCC(Q_1,\ldots,Q_n)}{\rightarrow} \sigma.$$ 
	\\
	From Rem.~\ref{rem:generalization} (point 2.) it follows that  $$\rho \stackrel{LOP(Q^{i(1)} \CI W^{j(1)} \ldots)}{\rightarrow} \sigma$$ and therefore $$\tau_{W} \stackrel{LOP(Q^{i(1)} \CI W^{j(1)} \ldots)}{\rightarrow} \sigma,$$ a contradiction.
\end{proof}

\section{Different types of entanglement}
\begin{ex.}
	The least coherent state necessary to produce the $\ket{GHZ_n}=\frac{1}{\sqrt{2}}(\ket{0}^{\otimes n}+\ket{1}^{\otimes n})$ state by $LOP(W \IC Q^{1}\otimes\ldots Q^n)$ is given by $\ket{+}_W=\frac{1}{\sqrt{2}}(\ket{0}+\ket{1})$. Similarly, one can produce $\ket{ {\cal W}_n}=\frac{1}{\sqrt{n}} (\ket{0,\ldots0, 1}+\ldots\ket{1,0,\ldots,0}$ using $LOP(W \IC Q^{1} \otimes\ldots Q^n)$ from $\ket{ +_{\log_2(n)}}_W=\frac{1}{\sqrt{n}}(\ket{0}+\ldots+\ket{n-1})$.
	
	The least coherent state necessary to produce the $\ket{GHZ_n}_{Q^1,\ldots,Q^n}$ state by $LOP(Q^{1}\CI  W^1 \IC Q^{2} \CI \ldots W^{n-1} \IC Q^n)$ is given by $\ket{+_{n-1}}_{W^1,\ldots W^{n-1}}=(\frac{1}{\sqrt{2}}(\ket{0}+\ket{1}))^{\otimes (n-1)}$. Similarly, one can produce $\ket{ {\cal W}_3}$ by $\ket{+_{\cal W}}_{W^1, W^2}:=\frac{1}{\sqrt{2}}(\ket{0}+\ket{1}) \otimes \frac{1}{\sqrt{3}}(\ket{0}+\sqrt{2}\ket{1})$, using $LOP(Q^{1}\CI  W^1 \IC Q^{2} \CI  W^{2} \IC Q^3)$.
\end{ex.}
\begin{proof}
	We start by giving explicit protocols that do the conversions. For the $\ket{GHZ_n}_{Q^1,\ldots,Q^n}$ state, we simply apply a CNOT (a permutation) on $\ket{+}_{W_1} \otimes \ket{0}^{(n-1)}_{W_2,\ldots,W_{n-1}}$ $n-1$ times (on ${W_1,W_j}$), resulting in $\ket{GHZ_n}_W$, then teleport the respective subsystems. For the $\ket{\cal W}$ state we apply a permutation that brings $\ket{i}_{W_1}\otimes\ket{0}^{(n-1)}_{W_2,\ldots,W_{n-1}}$ to $\ket{0}_{W_1}\ldots \ket{0} \ket{1}_{W_i}\ket{0}\ldots\ket{0}_{W_{n-1}}$ to $\ket{ +_{\log_2(n)}}_{W_1}\otimes\ket{0}^{(n-1)}_{W_2,\ldots,W_{n-1}}=\frac{1}{\sqrt{n}}(\ket{0}+\ldots+\ket{n-1})\otimes\ket{0}^{(n-1)}_{W_2,\ldots,W_{n-1}}$ and teleport the subsystems to the different parties. That the generation of the $\ket{GHZ_n}$ state is optimal is simply seen by the fact that $\ket{GHZ_2}$ is the state with the minimal coherence rank having the relative entropy of coherence equal to the relative entropy of entanglement of $\ket{GHZ_n}_{Q^1,\ldots,Q^n}$, namely $1$.
	
	For the case with more than one wire, to prepare the $\ket{GHZ_n}_{Q^1,\ldots,Q^n}$ state from $\ket{+_{n-1}}_{W^1,\ldots,W^{n-1}}$, we do a local CNOT on each wire, after preparing an ancillary state $\ket{0}W^j_2$, effectively "doubling" the states, resulting in $\ket{GHZ_{2}}^{\otimes(n-1)}_{(W^1_1,W^1_2),\ldots,(W^{n-1}_1,W^{n-1}_2)}$. We then forward each half of the system of the respective wires to the quantum systems they connect, resulting in $\ket{GHZ_{2}}^{\otimes(n-1)}_{(Q^1_1,Q^2_1),(Q^2_2,Q^3_1),\ldots,Q^{n-1}_1,(Q^{n-1}_2,Q^{n}_1)}$. We can then double the system $Q^2_1$, resulting in the state $\ket{GHZ_{3}}_{Q^1_1,Q^2_1,Q^2_3}$. We can then use the $\ket{GHZ_{2}}_{(Q^2_2,Q^3_1)}$ to teleport the system $Q^2_3$ to $Q^3_3$, and so on.  Iteratively we get the wanted $\ket{GHZ_n}_{Q^1,\ldots,Q^n}$ state. The optimality follows from any bipartitions being equivalent to $\ket{GHZ_2}$ states.
	
	Starting from the $\ket{+_{\cal W}}_{W^1, W^2}$ state, we first "double" each of the sides, then make a permutation $1\leftrightarrow 0$ on the site $W^2_2$ which puts the system into the state:	$\ket{GHZ_2}_{W^1_1,W^1_2}\otimes\frac{\sqrt{2}}{\sqrt{3}}(\ket{0,1}+\frac{1}{\sqrt{3}}\ket{1,0})_{W^2_1,W^2_2}$. We continue by forwarding the system $W^2_2$ to $Q^2_1$. On this site ($Q^2$) we then continue by applying the operation $\ket{0,0}_{Q^2_1,Q^2_2}\bra{0}_{Q^2_1}+\frac{1}{\sqrt{2}}\ket{0,1}_{Q^2_1,Q^2_2}\bra{1}_{Q^2_1}+\frac{1}{\sqrt{2}}\
	\ket{1,0}_{Q^2_1,Q^2_2}\bra{1}_{Q^2_1}$ leaving us with the $\ket{{\cal W}}$ state on $Q^2_1\otimes Q^2_2 \otimes W^2_1$. The next step is to forward the system $W^2_1$ to $Q^3$. Finally one can distribute the $\ket{GHZ_2}_{W^1_1,W^1_2}$ state to the connected quantum systems (yielding  $\ket{GHZ_2}_{Q^1_1,Q^2_3}$), and use this to teleport via LOCC and hence LOP the system $Q^2_2$ to $Q^1_2$. The protocol hence results in $\ket{{\cal W}}_{Q^1_2,Q^2_1,Q^3}$.	
\end{proof}

\section{All not incoherent-quantum states are resource states}\label{sec:resources}
In this section we show by a very simple (but highly inefficient) protocol that any state that does not have the form $\sum_i \ketbra{i}{i}_W \otimes \rho(i)_Q$ is maximally useful, in the sense that having enough such states as ancillae one can do any operation. If a state does not have this form, it must have the form 
$\sum_{i,j} \ketbra{i}{j}_W \otimes\rho(i,j)_Q$, with $\rho(i_0,j_0)\neq 0$ for some $i_0,j_0$ (wlog $i_0=1$, $j_0=2$). The first step of the protocol is to double the state to get: $\sum_{i,j} \ketbra{i}{j}_W \otimes \ketbra{i}{j}_{Q_1}\otimes\rho(i,j)_{Q_2}$. To simplify the analysis we now note that we can make the measurement $K_1=1/\sqrt{2} (\bra{1}+\bra{2})$, $K_2=\sqrt{\id - K_1\mdag K_1}$ on $Q_1$, which with non-zero probability will result in a state proportional to $\sum_{i,j \in \{1,2\}} \ketbra{i}{j}_W \otimes\rho(i,j)_{Q_2}$. That means that as long as we are not interested in the optimal distillation protocol we can start wlog with the latter state. The next thing to note is that for a state of this form there is always a measurement on $Q_2$ that will yield some state with coherence on $W$ with non-zero probability (the algorithm is given in the proof of Thm. 2 in~\cite{Chitambar2016a}). That means that we can start wlog with a state $\sum_{i,j \in \{1,2\}} \ketbra{i}{j}_W \sigma(i,j)$, with $\sigma(1,2)\neq 0$. By a rotation we get $\sigma(1,2)> 0$ and by the map which is given by the identity with probability $1/2$ and the permutation $1\leftrightarrow 2$ with probability $1/2$, we can assume $\sigma(1,1)=\sigma(2,2)=1/2$. Having a state $\sigma_1(1,2)=p/2$ on $W_1$ and adding a second state with $\sigma_2(1,2)=q/2>0$ $W_2$ one can first do a CNOT with control $W_1$ acting on $W_2$ followed by the measurement ($\bra{+}_{W_1}=(\bra{1}_{W_1}+\bra{2}_{W_1})/\sqrt{2}$, $\bra{-}_{W_1}=(\bra{1}_{W_1}-\bra{2}_{W_1})/\sqrt{2}$). With probability $\frac{1}{2} (1+p q)$ this yields the result ``$+"$ and the state
\[\left(
\begin{array}{cc}
\frac{1}{2} & \frac{p+q}{2 p q+2} \\
\frac{p+q}{2 p q+2} & \frac{1}{2} \\
\end{array}
\right)\]
and with probability $\frac{1}{2} (1-p q)$ one gets the result ``$-"$ and the state
\[\left(
\begin{array}{cc}
\frac{1}{2} & \frac{p-q}{2-2 p q} \\
\frac{p-q}{2-2 p q} & \frac{1}{2} \\
\end{array}
\right).\]

\begin{figure}[h]
	\centering
	\includegraphics[width=9cm]{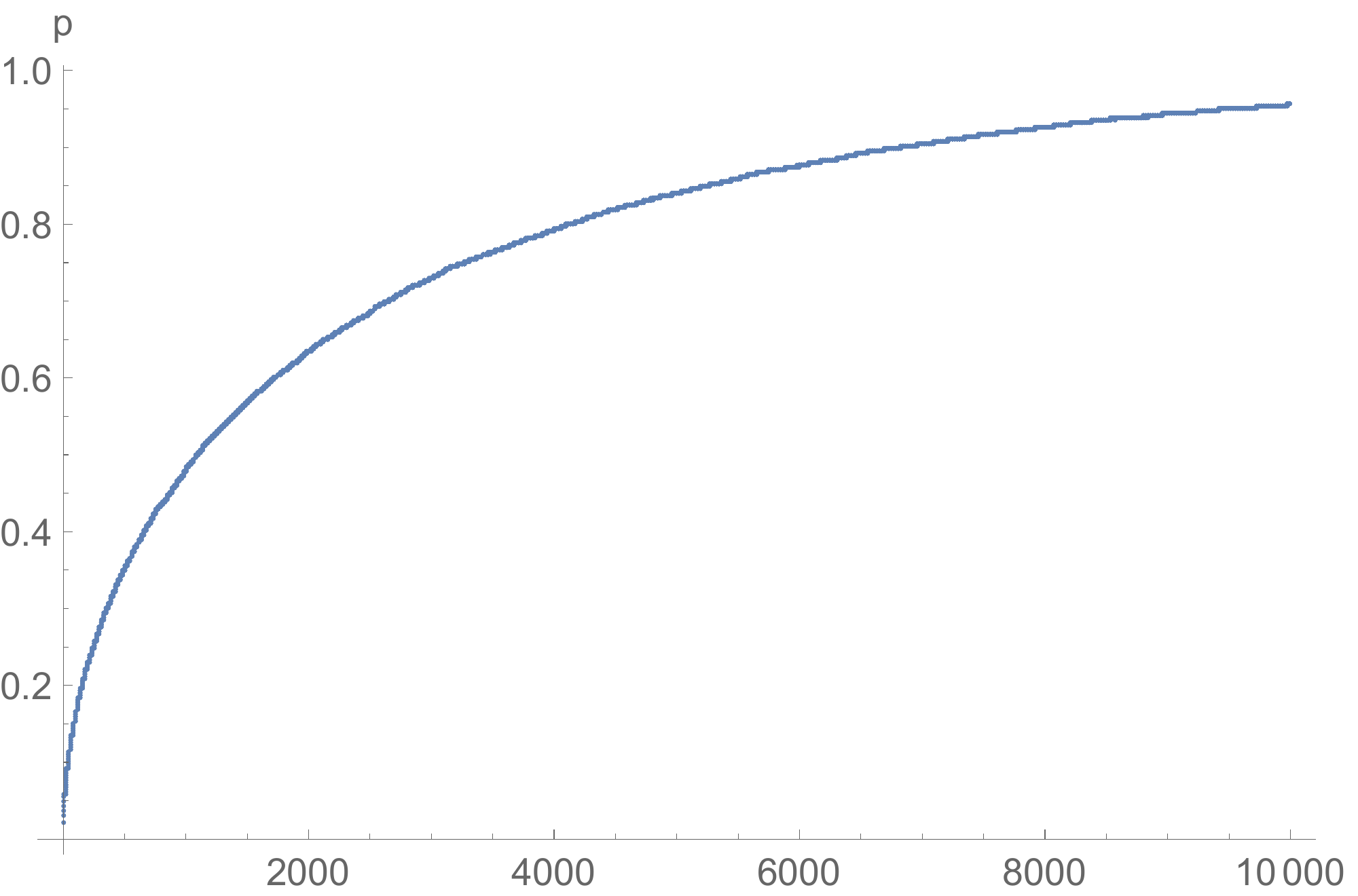}
	\caption{The figure shows a simulation of the protocol discussed in Appendix~\ref{sec:resources}. repeatedly a state with l1-coherence $0.01$ is used to increase the l1-coherence $p/2$ of a given state (also initialized with coherence $0.01$). If the coherence gets negative the state is dropped. This is repeated 10000 times to get the average behaviour. One sees that at some point the protocol saturates with coherence $p/2\approx0.5$}
	\label{fig:distillationofcoherence}
\end{figure}

Notice that repeating the sequence many times (adding many times a state with coherence $q$), if one gets symmetric outcomes (same number of + and -) the resulting state is equal to the initial one, while the probability to get outcomes with more $+$ is higher than to get outcomes with more $-$. There is therefore a bias to get more coherent states over time. Also the protocol only saturates at $p_{final}=1$. This means that for a given $\epsilon$, having enough copies of any state which is not incoherent-quantum, LOP allows to prepare a maximally coherent state with fidelity $f>1-\epsilon$ and probability $p> 1-\epsilon$. A simulation of this protocol is given in Fig.~\ref{fig:distillationofcoherence}. Of course the protocol is sub-optimal in many ways: it only considers two-dimensional coherence and destroys even a part of that by making the state symmetric at the beginning. Furthermore one could maybe improve the algorithm by grouping states. In any case it shows that one can distillate coherence in this setting, similar as it is done in~\cite{Winter2016} using incoherent operations, even though the specific protocol given there does not seem to be easily adapted to LOP. It remains an interesting open question what is the best possible distillation rate of coherence in LOP.

\section{Counterexample}
\label{sec:counterexample}

\begin{lem}\label{lem:rank1}
	Be ${\bf A}$ a rank 1 incoherent Kraus operator, and $\{{\bf B}_s\}$ a set of incoherent Kraus operators such that
	$$
	{\bf A}^\dagger{\bf A}=\sum_s {\bf B}_s^\dagger {\bf B}_s
	$$
	then, ${\bf B}_s =\lambda_s {\bf U}_s {\bf A}$ for certain $\lambda_s >0$ and ${\bf U}_s$ incoherent unitary operations.	
\end{lem}

Proof:  Since ${\bf A}^\dagger{\bf A}$ is rank 1,  and ${\bf B}_s^\dagger{\bf B}_s$ are all positive, then ${\bf B}_s^\dagger{\bf B}_s = \lambda_s^2 {\bf A}^\dagger{\bf A}$ for certain $\lambda_s$, such that $\sum_s \lambda_s^2 =1$. Therefore, from the singular value decomposition  theorem, ${\bf A}^\dagger{\bf A}=({\bf B}_s/\lambda_s)^\dagger({\bf B}_s/\lambda_s)$ iff ${\bf B}_s/\lambda_s = {\bf U}_s {\bf A}$ for certain unitaries ${\bf U}_s$. Finally, because ${\bf A}$ and ${\bf B}_s$ are incoherent,   ${\bf U}_s$ must be incoherent.

\subsection{Lemma: general form for 3 level systems}

From theorem  \ref{thm:conv}, the most general form of an operation $\Lambda\in LOP(W\IC Q)$, with $\dim(W)=3$, $\dim(Q)=1$ is given by
\begin{eqnarray}
\Lambda(\rho)&=& \sum_{\alpha_1 \in I_{1,3}} {\bf F}_{1}^{\alpha_1}\rho {\bf F}_{1}^{\alpha_1\dagger} + 
\mathop{\sum_{\alpha_2\in I_{2,2}}}_{\alpha_1\in I_{1,2}}
{\bf F}_{2,2}^{\alpha_2\alpha_1}\rho {\bf F}_{2,2}^{\alpha_2 \alpha_1 \dagger}+ \nonumber\\
&&
\mathop{\sum_{\alpha_2\in I_{2,1}}}_{\alpha_1\in I_{1,1}}
{\bf F}_{2,1}^{\alpha_2\alpha_1}\rho {\bf F}_{2,1}^{\alpha_2 \alpha_1\dagger}+
\mathop{\mathop{\sum_{\alpha_3\in I_{3,1}}}_{\alpha_2\in I'_{2,2}}}_{\alpha_1\in I_{1,2}}{\bf F}_{3}^{\alpha_3\alpha_2\alpha_1}\rho {\bf F}_{3}^{\alpha_3\alpha_2 \alpha_1\dagger}
\end{eqnarray}$$
$$
with (defining for notational ease $\cut_t(x):= \min(t,x) $)
\begin{subequations}
	\begin{align}
	{\bf F}_1^{\alpha_1} &= \sum_{m} q_m^{\alpha_1} |\sigma^{\alpha_1}(m)\rangle\langle m| \;  \label{3lform:f1} \\
	{\bf F}_{2,2}^{\alpha_2\alpha_1} &= \sum_{m} q_{m}^{\alpha_2\alpha_1} |(\sigma^{\alpha_2\alpha_1}\circ {\rm cut}_2\circ \sigma^{\alpha_1})(m)\rangle\langle m|  \label{3lform:f2} \\
	{\bf F}_{2,1}^{\alpha_2\alpha_1} &= \sum_{m} q_{m}^{\alpha_2\alpha_1} |(\sigma^{\alpha_2\alpha_1}\circ {\rm cut}_1\circ \sigma^{\alpha_1})(m)\rangle\langle m|  \label{3lform:f3}\\
	{\bf F}_3^{\alpha_3\alpha_2\alpha_1} &= \sum_{m} q_{m}^{\alpha_3\alpha_2\alpha_1} |(\sigma^{\alpha_3\alpha_2\alpha_1}\circ {\rm cut}_1
	\circ
	\sigma^{\alpha_2\alpha_1}\circ {\rm cut}_2 \circ  \sigma^{\alpha_1})(m)\rangle\langle m|   \label{3lform:f4}
	\end{align}
\end{subequations}
and $\alpha_1 \in I_{1,3} \cup I_{1,2} \cup I_{1,1}$, $\alpha_2 \in I_{2,1} \cup I_{2,2} \cup I'_{2,2}$ and $\alpha_3 \in I_{3,1}$. Due to the trace preserving condition, these operators must satisfy the constraints
\begin{subequations}
	\label{eq:TPC}
	\begin{align}
	\sum_{\alpha_2 \in I_{2,1}}{\bf F}_{2,1}^{\alpha_2\alpha_1\dagger}{\bf F}_{2,1}^{\alpha_2\alpha_1} &=  {\bf F}_{1}^{\alpha_1\dagger }{\bf F}_{1}^{\alpha_1}
	\;\;, \alpha_1 \in I_{1,1}
	\\
	\sum_{\alpha_2 \in I_{2,2}}{\bf F}_{2,2}^{\alpha_2\alpha_1\dagger}{\bf F}_{2,2}^{\alpha_2\alpha_1} &=   {\bf F}_{1}^{\alpha_1\dagger }{\bf F}_{1}^{\alpha_1}
	\;\;, \alpha_1 \in I_{1,2}\\
	\sum_{\alpha_3 \in I_{3,1}}{\bf F}_{3}^{\alpha_3\alpha_2\alpha_1\dagger}{\bf F}_{3}^{\alpha_3\alpha_2\alpha_1} &=  {\bf F}_{2,2}^{\alpha_2\alpha_1\dagger }{\bf F}_{2,2}^{\alpha_2\alpha_1}\;\;, \alpha_1 \in I_{1,2}, \alpha_{2} \in  I'_{2,2}\\
	\sum_{\alpha_2 \in I_{2,2} \cup I'_{2,2}}{\bf F}_{2,2}^{\alpha_2\alpha_1\dagger}{\bf F}_{2,2}^{\alpha_2\alpha_1} &=  {\bf F}_{1}^{\alpha_1\dagger }{\bf F}_{1}^{\alpha_1}
	\;\;, \alpha_1 \in I_{1,2}\\
	\sum_{\alpha_1 \in I_{1,1} \cup I_{1,2}\cup I_{1,3}}{\bf F}_{1}^{\alpha_1\dagger}{\bf F}_{1}^{\alpha_1} &=  {\bf 1}.
	\end{align}
\end{subequations}

\emph{Proof}: It follows from Thm.~\ref{thm:conv}, if we assume without loss of generality, that the initial global state is a product of the initial state $\rho$ on $W$, and reference ancillary state $\rho_{Q}=|0\rangle\langle 0|_{Q}$. The final state is then given by
$$
\Lambda(\rho)=  {\rm Tr}_Q\Lambda_{WQ}(\rho\otimes |0\rangle\langle 0|_Q) = \sum_{\vec{\alpha}}\sum_{m,m'} Q^{\vec{\alpha}}_{mm'} |f^{\vec{\alpha}}(m)\rangle\langle m |\rho
|m'\rangle\langle f^{\vec{\alpha}}(m') |
$$
being $f^{\vec{\alpha}}=\sigma^{\alpha_k}\circ cut^{r_{\alpha_k}\ldots \alpha_1}\circ \ldots \circ \sigma^{\alpha_1}$  and  $Q^{\vec{\alpha}}_{mm'}={\rm Tr}({\bf E}^{\vec{\alpha}}_{m}|0\rangle\langle 0|{\bf E}^{\vec{\alpha}\dagger}_{m'})=\langle 0|{\bf E}^{\vec{\alpha}\dagger}_{m'}{\bf E}^{\vec{\alpha}}_{m}|0\rangle=Q^{\vec{\alpha}\dagger}_{mm'}$ with 
${\bf E}_{m}^{\vec{\alpha}}= {\bf E}_{f^{\alpha_k\ldots\alpha_1}(m)}^{\alpha_k\ldots\alpha_1}\ldots {\bf E}_{f^{\alpha_1}(m)}^{\alpha_2\alpha_1} {\bf E}_m^{\alpha_1}$ the sequence of conditional operators applied on each step. By construction, $Q^{\vec{\alpha}}_{mm'}$ is a Grahm matrix, and hence, it is positive semidefinite. With $\zeta_m^{\vec{\alpha}\lambda}=\bra{\lambda}{\bf E}^{\vec{\alpha}}_{m}\ket{0}$, it can be 
expended as  $Q_{mm'}^\alpha= \sum_{\lambda} \zeta_m^{\vec{\alpha}\lambda} \zeta_{m'}^{\vec{\alpha}\lambda\dagger}$ and we see that the general form of $\Lambda$ is given by
$$
\Lambda(\rho)=\sum_{\vec{\alpha},\lambda}  \tilde{\bf F}^{\vec{\alpha}\lambda} \rho \tilde{\bf F}^{\vec{\alpha}\lambda\dagger}
$$
with $\tilde{\bf F}^{\vec{\alpha}\lambda}= \sum_{m} \zeta^{\vec{\alpha}\lambda}_{m}|f_{m}^{\vec{\alpha}}\rangle\langle m |$. The $\lambda$ coefficient in the sum can be assimilated to the last set of outcomes, leading to the form presented in the theorem.

\subsection{Proposition:  \texorpdfstring{$LOP \neq IQO$}{LOP != IQO}}
Suppose now that the Incoherent operation  $\Lambda(\rho)=\sum_{s=1}^4{\bf K}_s\rho {\bf K}^\dagger_s$ defined by

$$
{\bf K}_1=\left(
\begin{array}{ccc}
\frac{1}{2} & -\frac{1}{2} & 0 \\
0 & 0 & \frac{1}{2} \\
0 & 0 & 0 \\
\end{array}
\right) \hspace{2cm}
{\bf K}_2=\left(
\begin{array}{ccc}
\frac{1}{2} & 0 & -\frac{1}{2} \\
0 & \frac{1}{2} & 0 \\
0 & 0 & 0 \\
\end{array}
\right)
$$

$$
{\bf K}_3=\left(
\begin{array}{ccc}
0 & \frac{1}{2} & -\frac{1}{2} \\
\frac{1}{2} & 0 & 0 \\
0 & 0 & 0 \\
\end{array}
\right)
\hspace{2cm}
{\bf K}_4=\left(
\begin{array}{ccc}
\frac{1}{2} & \frac{1}{2} & \frac{1}{2} \\
0 & 0 & 0 \\
0 & 0 & 0 \\
\end{array}
\right)
$$
admits that decomposition. Since $\sum_\mu \lambda_\mu {\bf K}_{s_\mu}$ is an  incoherent Kraus operator iff $s_\mu = {\rm cst}$  it follows that
in any possible incoherent Kraus representation of  $\Lambda(\rho)$, all the Kraus operators need to be proportional to some  ${\bf K}_s$, hence,
\begin{eqnarray}
{\bf F}_1^{\alpha_1} &=& 0\;\; \forall \alpha_1\in I_{1,3} \\
{\bf F}_{2,2}^{\alpha_2\alpha_1} &=& \zeta_{2,2}^{\alpha_2,\alpha_1} {\bf K}_{m_{\alpha_2,\alpha_1}} \;\; m_{\alpha_2,\alpha_1} \in \{1,2,3\}   \\
{\bf F}_{2,1}^{\alpha_2\alpha_1} &=& \zeta_{2,1}^{\alpha_2,\alpha_1} {\bf K}_4   \\
{\bf F}_3^{\alpha_3\alpha_2\alpha_1} &=& \zeta_{3}^{\alpha_3,\alpha_2,\alpha_1} {\bf K}_4 
\end{eqnarray}
Plugging this in the trace preserving conditions \ref{eq:TPC} result in
\begin{eqnarray}
{\bf K}_4^\dagger {\bf K}_4\left(\sum_{\alpha_3\in I_{3,1}}|\zeta_3^{\alpha_3 \alpha_2\alpha_1}|^2\right)&=& {\bf F}_{2,2}^{\alpha_2\alpha_1\dagger}{\bf F}_{2,2}^{\alpha_2\alpha_1}\\
{\bf K}_4^\dagger {\bf K}_4\left(\sum_{\alpha_{2,1}}|\zeta_{2,1}^{\alpha_2\alpha_1}|^2\right)&=& {\bf F}_{1}^{\alpha_1\dagger}{\bf F}_{1}^{\alpha_1}
\end{eqnarray}
If the left hand side is non-zero in the first (second) condition, it implies that for certain $\alpha_1,\alpha_2$  ($\alpha_1$)  ${\bf F}_{2,2}^{\alpha_2\alpha_1}$ (${\bf F}_{1}^{\alpha_1}$)  should (by Lem.~\ref{lem:rank1}) be proportional to ${\bf U}{\bf K}_4$ for certain ${\bf U}$ unitary incoherent, but then, 
against the hypothesis,  ${\bf F}_{2,2}^{\alpha_2\alpha_1}$ (${\bf F}_{1}^{\alpha_1}$) needs to be of the form \ref{3lform:f2} (\ref{3lform:f1}).
On the other hand, if both expressions are zero, ${\bf K}_4$ can not appear in the Kraus decomposition of $\Lambda$, leading to a contradiction. Therefore the explicitly incoherent operation $\Lambda$ is not in $LOP(W\IC Q)$, while being an element of $IQO$.

\end{document}